\documentclass[aps, prd, reprint, nofootinbib, longbibliography,floatfix, superscriptaddress]{revtex4-2}
\usepackage{natbib}
\usepackage[utf8]{inputenc}
\usepackage{float}
\usepackage{soul}
\usepackage{cancel}
\usepackage{subfigure}
\usepackage{amssymb,amsmath,amsfonts,amsthm,epsfig,epstopdf,array,mathrsfs}
\usepackage{braket}
\usepackage{graphicx}
\usepackage{dcolumn}
\usepackage{ marvosym }

\usepackage{bm}% bold math
\usepackage{hyperref}% add hypertext capabilities
\usepackage{cleveref}
\usepackage{tikz}
\usepackage{comment}

\def\bea{\begin{eqnarray}}
\def\eea{\end{eqnarray}}
\def\be{\begin{equation}}
\def\ee{\end{equation}}

\theoremstyle{definition}

\newtheorem{defn}{Definition}
\newtheorem{prop}{Proposition}

\definecolor{mrainbow1}{rgb}{0.86, 0.13, 0.13}
\definecolor{mrainbow2}{rgb}{0.89,0.60,0.22}
\definecolor{mrainbow3}{rgb}{0.67, 0.74, 0.32}
\definecolor{mrainbow4}{rgb}{0.39, 0.67, 0.60}
\definecolor{mrainbow5}{rgb}{0.25, 0.39, 0.81}
\definecolor{mrainbow6}{rgb}{0.47, 0.11, 0.53}
\begin{document}

\preprint{}

\title{Inflation does not create entanglement in local observables}% Force line breaks with \\
%\thanks{}%
%Entanglemetn from in inflation in real space

\author{Patricia Ribes-Metidieri}
\email{patricia.ribesmetidieri@ru.nl}

\affiliation{Institute for Mathematics, Astrophysics and Particle Physics, Radboud University, 6525 AJ Nijmegen, The Netherlands}
%\affiliation{Department of Physics and Astronomy, Louisiana State University, Baton Rouge, LA 70803, USA}

\author{Ivan Agullo}
\email{agullo@lsu.edu}

\affiliation{Department of Physics and Astronomy, Louisiana State University, Baton Rouge, LA 70803, USA}
%\affiliation{Perimeter Institute for Theoretical Physics, Waterloo, Ontario, N2L 2Y5, Canada}

\author{B\'eatrice Bonga}
\email{bbonga@science.ru.nl}

\affiliation{Institute for Mathematics, Astrophysics and Particle Physics, Radboud University, 6525 AJ Nijmegen, The Netherlands}

\date{\today}% It is always \today, today,
             %  but any date may be explicitly specified

\begin{abstract}
Using modern tools of relativistic quantum information, we compare entanglement of a free, massive scalar field in the Bunch-Davies vacuum in the cosmological patch of de Sitter spacetime with that in Minkowski spacetime. There is less entanglement between spatially localized field modes in de Sitter, despite the fact that there is more entanglement stored in the field on large scales. 
This shows that inflation does not produce entanglement between local observables.
\end{abstract}
             
\maketitle
{\bf Introduction.} It has long been argued that cosmic inflation squeezes cosmological perturbations \cite{Grishchuk:1990bj,Albrecht:1992kf,Lesgourgues:1996jc,Kiefer:1998pb,Kiefer:1998qe,Kiefer:2008ku,Polarski:1995jg,Martin:2015qta,ack,Micheli:2022tld,Brahma:2023lqm,Brahma:2023uab,Brahma:2023hki,Brahma:2024yor,Bhattacharyya:2024duw,Martin_2022}, and thereby generates entanglement between perturbations with opposite wave numbers, $\vec{k}$ and $-\vec{k}$. Detecting any trace of this entanglement would confirm one of the pillars of modern early-universe cosmology: the quantum origin of cosmological perturbations. There has been a recent surge of interest in applying quantum information tools to this problem, both to quantify the entanglement generated and to identify sources of decoherence  during the post-inflationary  evolution \cite{Calzetta_1995,maldacena_entanglement_2013,Maldacena:2015bha,Martin:2015qta,Nelson_2016,ack,Martin_2018,Grain:2019vnq,Brahma:2020zpk,Brahma:2021mng,Martin:2021qkg,Agullo:2022ttg,Micheli:2022tld,Espinosa-Portales:2022yok,Bhardwaj:2023squ}. Various  arguments show that a significant portion of the primordial entanglement is likely to decohere (see, e.g., \cite{Calzetta_1995,Lombardo_2005,Bhattacharyya:2024duw,Burgess_2008,burgess2014efthorizonstochasticinflation,Burgess_2023,Brahma:2024yor}). Even if some portion remains, current cosmic microwave background (CMB) observations might be insufficient to detect it \cite{Maldacena:2015bha}. However, the possibility of observing entanglement is not entirely ruled out, particularly if primordial gravitational waves are eventually observed \cite{Micheli:2022tld}.

We critically review the claim that inflation imprints entanglement in cosmological observations.
%using modern tools of relativistic quantum information. 
Previous discussions have primarily focused on Fourier space (with important exceptions \cite{Martin:2021xml,Martin:2021qkg,K:2023oon}), but Fourier modes are inherently global, meaning that certain aspects of them are not accessible to local observers.

For entanglement, it is crucial to move away from Fourier space and work in real space, because localizing a mode in quantum field theory inevitably results in a mixed reduced density operator --- due to the correlations with other modes --- and this mixedness acts as a source of decoherence. Therefore, to quantify the entanglement generated by inflation accessible to us, it is essential to localize field observables within our Hubble horizon. This requires moving beyond the description in terms of uncoupled and uncorrelated Fourier modes.

Although the effects of decoherence during post-inflationary evolution and the limitations of our observational apparatuses are crucial for gauging the detectability of entanglement, we focus on a more fundamental aspect: how much entanglement is created during inflation in the first place.
%, in field modes that will eventually become accessible to us. 
To answer this, we will evaluate the entanglement present in the quantum state of cosmological perturbations at the end of inflation, before the universe reheats and interactions with other fields become significant.  
Given the absence of non-Gaussianities in the CMB on observable scales \cite{refId0}, we restrict to linear field theory.

%Quantifying entanglement in absolute terms is futile for our goal. Instead, it is more informative 
We compare the entanglement in de Sitter with that in flat space-time---the difference quantifies how much entanglement can be attributed to the inflationary expansion. The comparison is meaningful because equal-time slices in the cosmological patch of de Sitter are isometric to equal inertial time  slices of flat spacetime, allowing us to identify corresponding modes in the two spacetimes. %This comparison will help us quantifying how much entanglement can be attributed to the inflationary expansion.

{\bf Approach: symplectic invariants.} The von Neumann entropy associated with a region has predominantly been used to study entanglement in field theory. While a suitably regularized version of this quantity provides information about the entanglement between a region and its complement, von Neumann entropy is not well-suited for quantifying entanglement between localized sets of modes, because their reduced state is mixed. 
%For instance, entropy would be nonzero even if two modes are unentangled, due to correlations with third modes.

A complementary and powerful approach that has recently emerged (see, e.g., \cite{bianchi_entropy_2019,Martin:2021xml,Agullo:2023fnp,Perche:2023nde,Perche:2023lwo}) involves defining finite-dimensional subsystems by smearing the field operator with functions of compact support, and applying tools from Gaussian quantum information theory to quantify entanglement in these finite-dimensional systems. This strategy is particularly useful for free field theories and Gaussian states, which is the scenario considered here.

We use a free, real scalar field in the cosmological patch of de Sitter space-time to model scalar cosmological perturbations during inflation. The field has a non-vanishing mass $m$, which helps to control infrared divergences. The regime relevant for inflation is $m/H \ll 1$, where $H$ denotes the Hubble rate.

We define the vector $\hat {\bf R}(\vec x)=(\hat \Phi(\vec x), \hat \Pi(\vec x))$, where $\hat \Phi(\vec x)$ and $\hat \Pi(\vec x)$ denote the field operator and its conjugate momentum, respectively, satisfying canonical commutation relations $[\hat  R^i(\vec x), \hat R^j(\vec x')]=i\hbar \Omega^{ij}(\vec x,\vec x')$, where ${\bf \Omega}(\vec x,\vec x')=\begin{pmatrix} 0 & 1 \\ -1 & 0\end{pmatrix}\, \delta^{(3)}(\vec x-\vec x')$. 
%${\bm \omega}(\vec x,\vec x')$ is the inverse of  ${\bm \Omega}(\vec x,\vec x')$, which can be identified with the symplectic structure of the classical phase of space of the theory. 

Let $|0\rangle$ be the Bunch-Davies vacuum \cite{Bunch:1978yq}. This is a Gaussian state with zero ``mean'', $\langle 0| \hat R^i(\vec x)|0 \rangle=0$. Its two-point functions can be decomposed into its symmetric and anti-symmetric parts: $\langle \hat R^i(\vec x),\hat R^j(\vec x')\rangle=\frac{1}{2}\sigma^{ij}(\vec x,\vec x')+\frac{1}{2}\hbar \, \Omega^{ij}(\vec x-\vec x')$, with the anti-symmetric part state-independent. %, ${\bm \sigma}(\vec x,\vec x')$  encodes the state dependence. 
Gaussianity implies that all higher-order correlation functions can be obtained from ${\bm \sigma}$. Thus, ${\bm \sigma}$ encodes all the information of the state. 
%All quantities computed below will be obtained from it. 
%are encoded in the symmetric matrix-bi-distribution ${\bm \sigma}(\vec x,\vec x')$: covariance matrix. 

We consider Hermitian operators linear in $\hat \Phi(\vec x)$ and $\hat \Pi(\vec x)$. Each such operator can be labeled by an element of the classical phase space ${\bm \gamma}(\vec x)=(g(\vec x),f(\vec x))$ %{\color{green} We need to change either this definition to ${\bm \gamma}(\vec x)=(f(\vec x),g(\vec x))$ or the expressions of  the correlations below} 
via  \cite{Wald:1995yp}
\begin{equation}
\label{Ogamma} 
\hat O_{\gamma}
%&=&{\bm \omega}[{\bm \gamma},{\bf R}]
%:=
%\int_{\Sigma_{\eta}} d^3x\,d^3x'\,  \omega_{ij}(\vec x,\vec x')\, 
%\gamma^{i}(\vec x)\hat R^j(\vec x') \, \\
=\int_{\Sigma} d^3x \, \Big(f(\vec x)\, \hat \Phi(\vec x)-g(\vec x)\, \hat \Pi(\vec x)\Big)\, .
\end{equation}
%{\color{green} Add that $f$ is a density? Also, it may make more sense to compare with Minkowski if we separate the function from the $\sqrt{h}$.  }
%Labeling linear operators by elements of the phase space has a number of advantages. For instance, the commutator of two such operators is determined by the symplectic product of the corresponding phase space elements
%
%$$[\hat O_{\gamma},\hat O_{\gamma'}]=i\hbar \ {\bm \omega}[{\bm \gamma},{\bm \gamma'}].$$
%
%Consider a pair of noncommuting observables $(\hat O_{\gamma},\hat O_{\gamma'})$. They span a subalgebra of the field theory, which is isomorphic to that of a quantum harmonic oscillator. This 
%
We are interested in subsystems containing a single degree of freedom. 
%{\color{blue} Classically, such subsystems are in one-to-one correspondence with two-dimensional symplectic subspaces of the phase space.} 
In the algebraic approach to quantum field theory, subsystems are defined in terms of subalgebras: given a pair of canonically conjugated  observables $(\hat O_{\gamma^{(1)}},\hat O_{\gamma^{(2)}})$,  the subalgebra they span defines a single-mode subsystem. Notice that any pair of operators resulting from a linear symplectic transformation of $\hat O_{\gamma}^{(1)}$ and $\hat O_{\gamma^{(2)}}$ would span the same algebra, hence define the same single-mode subsystem. 
%defines a subsystem of the field theory containing with a single degree of freedom (a.k.a. single mode). Mathematically, the subsytem is defined by the subalgebra of the field theory these two observable span---which is isomorphic to that of a quantum harmonic oscillator.

Let A and B denote two localized and independent (i.e., commuting) single-mode subsystems. The localization of each mode is determined by the support of the two phase space elements defining them, $\gamma^{(i)}_A(\vec x)=(g_A^{(i)}(\vec x),f_A^{(i)}(\vec x))$, $i=1,2$, and similarly for B. %{\color{green} --We only assume these functions to be smooth and of  compact support. In the examples in Figs.~\ref{MI}-~\ref{fig:shellballmany}, we show that the smoothness condition is in fact not crucial.----OR-----To prove the propositions in this article, we assume that these functions are smooth and have compact support, although the examples in Figs.~ ~\ref{MI}-~\ref{fig:shellballmany} demonstrate that smoothness is not critical.} 
 We only assume these functions to be smooth and of compact support---although our results remain valid if smoothness is relaxed \cite{longpaper} (see also examples in Figs.~\ref{MI}–\ref{fig:shellballmany} below).

The reduced state $\hat \rho_{AB}$ obtained from the Bunch-Davies vacuum by tracing the rest of degrees of freedom is also a Gaussian state with zero mean. Its symmetric second moments are given in terms of the four operators 
$\hat{\bf R}_{AB} = (\hat O_{\gamma^{(1)}_A},\,   \hat O_{\gamma^{(2)}_A},\, \hat O_{\gamma^{(1)}_B},\,   \hat O_{\gamma^{(2)}_B})$
 as $\sigma_{AB}^{ij}=\langle 0|\hat  R^i_{AB}\hat R^j_{AB}+\hat  R^j_{AB}\hat R^i_{AB}|0\rangle$. 
The structure of ${\bm \sigma}_{AB}$ is
\be \label{sigmaAB} 
{\bm \sigma}_{AB}=\left(
\begin{matrix}
{\bm \sigma}_A & {\bm  C}\\ 
{\bm  C}^{\top} & {\bm  \sigma}_B \\
\end{matrix}\right),\ee 
where ${\bm \sigma}_A$ and ${\bm  \sigma}_B$ are $(2\times 2)$-matrices describing the symmetrized second moments of each single mode individually, while ${\bm C}$ is a $(2\times 2)$-matrix describing their correlations. %$\sigma_{AB}$ codifies the same information as $\hat \rho_{AB}$ 
%{\color{blue} All properties of the two-mode system will be determined from $\sigma_{AB}$.}

We are interested in computing the entropy of subsystems and the correlations between them. These quantities are properties associated with subsystems, and not to any choice of operators  $(\hat O_{\gamma^{(1)}_I},\hat O_{\gamma^{(2)}_I})$, $I=A,B$  within each subsystem.  
More precisely, these quantities are invariant under ``subsystem-local'' symplectic transformation. As such, they can be computed from the invariant scalars of ${\bm \sigma}_{AB}$, namely  ${\rm det} \, {\bm  \sigma}_A$, ${\rm det} \, {\bm \sigma}_B$, ${\rm det} \, {\bm  C}$ and ${\rm det}\, {\bm \sigma}_{AB}$. % are all invariant under subsystem-local symplectic transformation. 
The following six combinations of them will be particularly useful:
\bea \nu_I&\equiv&\sqrt{{\rm det} \, {\bm \sigma}_I},\  I=A,B, \nonumber\\   \nu_{\pm}^2&\equiv& (\Delta\pm\sqrt{ \Delta^2-4 \, {\rm det}\, {\bm \sigma}_{AB}})/2, \\  \nonumber  \tilde \nu_{\pm}^2&\equiv& (\tilde \Delta\pm\sqrt{ \tilde \Delta^2-4 \, {\rm det}\, {\bm  \sigma}_{AB}})/2, \eea 
where $\Delta={\rm det} \, {\bm  \sigma}_A+{\rm det}\,{\bm  \sigma}_B+2\, {\rm det} \,{\bm C}$ and 
$\tilde \Delta={\rm det} \,{\bm  \sigma}_A+{\rm det}\, {\bm  \sigma}_B-2 \, {\rm det}\, {\bm  C}$.
%\bea \Delta&=&{\rm det} \sigma_A+{\rm det} \sigma_B+2 {\rm det} C\, ,\\ \tilde \Delta&=&{\rm det} \sigma_A+{\rm det} \sigma_B-2 {\rm det} C\, ,\eea
%
%From Eqs.~\eqref{Ogamma} and \eqref{sigmaAB}, one can see that 
The calculation of these invariants boils down to computing symmetrized expectation values of the smeared operators $\hat{\Phi}[f_I]=\int_{\Sigma} f_I \hat{\Phi}$ and $\hat{\Pi}[g_I]=\int_{\Sigma} g_I \hat{\Pi}$. 

We are interested in field modes which, at the end of inflation, are supported on ``super-Hubble'' regions, i.e., regions of physical size $R\gg H^{-1}$, because these are the modes that become accessible in observations today---%because these modes are affected by the spacetime curvature and become accessible in observations today. 
the minimum primordial wavelength resolvable in the CMB, which is of the order of $10^4$ Mpc today, at the end of inflation was $>e^{50}$ times the Hubble radius for typical inflationary models \cite{Liddle_2003}.

In the regime of interest, namely $RH\gg1$ and $m/H\ll 1$, the symmetrized expectation values of $\hat{\Phi}[f_I]$ and $\hat{\Pi}[g_I]$ in the Bunch-Davies vacuum are
\begin{widetext} \bea\label{phiphi} %\begin{split}
         \braket{\{\hat{\Phi}[f_I],\hat{\Phi}[f_J]\}} &=&  \mathrm{Re}\left[(f_I|f_J)_{-\frac{1}{2}} \right]+ \frac{2^{2-2{\mu}^2}\pi\, (RH)^{2-2{\mu}^2}}{\cos^2(\pi {\mu}^2) \Gamma \left(- \frac{1}{2} + {\mu}^2\right)^2 } \,  \mathrm{Re}\left[(f_I|f_J)_{-\frac{3}{2}+{\mu}^2}\right] +\mathcal{O}({\mu}^2 \log (RH)) \,, \\ 
 \label{pipi}
       \braket{\{\hat{\Pi}[g_I],\hat{\Pi}[g_J]\}} &=& \mathrm{Re}\left[(g_{I}| g_J)_{\frac{1}{2}} \right]+ \mathcal{O}({\mu}^2\log(RH))\,, \\
\label{phipi}
       \braket{\{\hat{\Phi}[f_I],\hat{\Pi}[g_J]\}} &=& \frac{2^{1-{\mu}^2}\sqrt{\pi} \, (RH)^{1-{\mu}^2}}{\cos(\pi {\mu}^2) \Gamma \left(-\frac{1}{2} + {\mu}^2\right)}\,  \mathrm{Re}\left[(f_{I}| g_J)_{-\frac{1- {\mu}^2}{2} } \right]+\mathcal{O}({\mu}^2 (RH)^{1-{\mu}^2})\,,
\eea
\end{widetext}
where ${\mu}^2 = \frac{3}{2} - \sqrt{\frac{9}{4} - \frac{m^2}{H^2}}\ll 1$ and the Sobolev product of order $s$ is %{\color{green} Change Sobolev inner product to depend on $\vec{q}:= \vec{k} R$ instead of $\vec{k}$}
\begin{equation}
    (f|g)_s = \int \frac{d^3 q}{(2\pi)^3} |\vec{q}|^{2s} \tilde{f}(\vec{q})\, \bar{\tilde{g}}(\vec{q})\,,
\end{equation}
with $\vec{q} := R \vec{k}$ a dimensionless wave vector. The last term in Eqs.~\eqref{phiphi}, \eqref{pipi} and \eqref{phipi} is subdominant when $RH\gg1$ and $m/H\ll 1$. 

In the limit $H\to 0$ only the first term in \eqref{phiphi} and \eqref{pipi} survives, \eqref{phipi} vanishes,  and these expressions reduce to the smeared two-point functions in Minkowski spacetime \cite{Agullo:2023fnp}. When $H\neq 0$,  $\mathrm{Re}\left[(f_A,f_B)_{-\frac{3}{2}+ \mu^2}\right] \sim \Delta x_{AB}^{ -2  \mu^2}$, where  $\Delta x_{AB}$ denotes the physical distance between the regions of support of $f_A$ and $f_B$;  %{\color{green} Here, we are implicitly assuming that the regions $A$ and $B$ are spherical!}; 
this term is responsible for the characteristic almost scale-invariant correlations generated by inflation, and is infrared divergent in the limit $m\to 0$, accounting for the well-known infrared divergence of the Bunch-Davies vacuum.\\

{\bf Entropy.} 
The von Neumann entropy of a single mode, say A, can be computed as
%Nevertheless, as noted in~\cite{}, the entropy of a single mode is a quantifier of the entanglement between a single mode of the field and its partner~\cite{}. The von Neumann entropy of the reduced Gaussian state obtained from tracing out the single-mode subsystem $B (A)$,   $\hat{\rho}_{A(B)}$,  
\be \label{eq:vNEntropy_nu}
S(\nu_A) = \frac{\nu_A +1}{2} \log_2 \left(\frac{\nu_A +1}{2}  \right) -\frac{\nu_A -1}{2} \log_2 \left(\frac{\nu_A -1}{2}  \right) \,.
\ee 
%With these preliminaries, we present the first result of this work: 
\begin{prop}
     When $R H \gg 1$, the von Neumann entropy of any single-mode subsystem of a scalar field prepared in the Bunch-Davies vacuum is equal to or greater than that of the same mode in the Minkowski vacuum of flat spacetime. 
\end{prop}
%Although this result applies to any single-mode subsystem, for brevity we restrict to the regime of interest for this article, namely $R H \gg 1$. { \color{green} This comment is certainly true for the restricted dofs, but I'm not so sure for general degrees of freedom.}
\begin{proof} 
 Throughout this article, we compare the small mass limit in de Sitter, $m/H\ll 1$, with the massless limit in Minwkoski. The absence of a scale in Minkowski makes this comparison physically sound. Furthermore, in Minkowski spacetime entropies and correlations become insensitive to $m$ in the limit $m R\ll1$ \cite{Agullo:2022ttg}, implying that maintaining  $m\neq 0$ when comparing with flat spacetime will not change the results.

    $S(\nu_A)$  monotonically increases with $\nu_A$. Thus, it suffices to show that $\nu_A^2 $ is larger in the Bunch-Davies vacuum than in the Minkowski vacuum. %In other words, it suffices to show that $\nu_A^2 - (\nu_{A}^{\mathrm{Mink}})^2 \geq 0$ when $\mathfrak{R}_H \gg 1$.  
    As discussed above, mode A can be defined from two canonically conjugate operators $(\hat O_{\gamma_A^{(1)}},\hat O_{\gamma_A^{(2)}})$, with $\gamma_A^{(i)}=(g_A^{(i)},f_A^{(i)})$, $i=1,2$. 
    Using~\eqref{phiphi}-\eqref{phipi}, the leading contributions to $\nu_A^2 - (\nu_{A}^{\mathrm{Mink}})^2$ can be expressed as a polynomial in $RH$. For sufficiently large $RH$, the sign of $\nu_A^2 - (\nu_{A}^{\mathrm{Mink}})^2$ is determined from the coefficient of the leading power in $R H$. This coefficient depends on the choice of the smearing functions $f_A^{(1)}$ and $f_A^{(2)}$. There are three cases to consider: %one can write the leading contributions of  $\nu_I^2 - (\nu_{I}^{\mathrm{Mink}})^2$ as a polynomial in $\mathfrak{R}_H$ in this regime.  For sufficiently large $\mathfrak{R}_H$, the sign of $\nu_I^2 - (\nu_{I}^{\mathrm{Mink}})^2$ is determined by the sign of the coefficient accompanying the leading power in $\mathfrak{R}_H$. From Eq~\eqref{}, one can see that the leading order (and therefore, the associated coefficient) in the $\nu_I^2 - (\nu_{I}^{\mathrm{Mink}})^2$ expansion depends on the choice of the smearing functions  $f_I^{(1)}$ and  $f_I^{(2)}$  There are three cases that one must analyze: 
    (1)  $f_A^{(1)} \neq f_A^{(2)}$, with both functions different from zero; (2) $f_A^{(1)} = f_A^{(2)}\neq 0$; and (3) $f_A^{(1)}=0$, $f_A^{(2)} \neq 0$ (or vice-versa). %{\color{blue} It seems weird that $g_A^{(i)}$ does nor play any role on this separation of cases.}
    
    In case (1), we find %Eq.~\eqref{}, it follows that the leading power of $\mathfrak{R}_H$ in the expansion of $\nu_I^2 - (\nu_{I}^{\mathrm{Mink}})^2$  depends on the choice of the smearing functions  $f_I^{(1)}$ and  $f_A^{(2)}$.
        $\nu_A^2 -(\nu_{A}^{\mathrm{Mink}})^2= \mathfrak{a}_A\, (RH)^{4-4{\mu}^2}(1 + \mathcal{O}({\mu}^2)) +\mathcal{O}\big((RH)^{3-3{\mu}^2}\big)\,,$ 
    with $$\mathfrak{a}_I = ||f_A^{(1)}||_{-\frac{3}{2}+{\mu}^2}^2||f_A^{(2)}||_{-\frac{3}{2}+{\mu}^2}^2 - \mathrm{Re}(f_A^{(1)}|f_A^{(2)} )_{-\frac{3}{2}+{\mu}^2}^2\,.$$
    The Cauchy-Schwarz inequality satisfied by  Sobolev products 
        %Since homogeneous Sobolev spaces of order $|s| < 3/2$ are Hilbert spaces with inner product given in Eq.~\eqref{}~\cite{}, the Cauchy-Schwarz inequality 
 implies that $\mathfrak{a}_A > 0$. Hence, $ \nu_A^2 -(\nu_{A}^{\mathrm{Mink}})^2 > 0 $ when $R H\gg 1$. In case (2), we find that $ \nu_A^2 -(\nu_{A}^{\mathrm{Mink}})^2 = \mathfrak{b}_A \, (RH)^{2-2{\mu}^2} (1 + \mathcal{O}({\mu}^2)) + \mathcal{O}\big((R H)^{1-{\mu}^2}\big)\,,$ where 
    $$\mathfrak{b}_A = ||f_A||^2_{-\frac{3}{2}+{\mu}^2} ||g_A^{(1)}-g_A^{(2)}||_{\frac{1}{2}}^2 - \mathrm{Re} (f_A|g_A^{(1)}-g_A^{(2)})_{-\frac{1-{\mu}^2}{2}}^2\,.$$ 
    Applying H\"older's inequality (see, e.g.~\cite{bahouri_fourier_2011})  for $s'\in [-1,1]$ and $s=-(1- {\mu}^2)/2$ , we find $\mathfrak{b}_A \geq 0$. In case (3) we find $\nu_A^2 -(\nu_{A}^{\mathrm{Mink}})^2= \mathfrak{c}_A \, (R H)^{2-2{\mu}^2}(1 + \mathcal{O}({\mu}^2)) + \mathcal{O}\big((RH)^{1-{\mu}^2}\big)$, with $$ \mathfrak{c}_I = ||f_A^{(2)}||^2_{-\frac{3}{2}+{\mu}^2} ||g_A^{(1)}||_{\frac{1}{2}}^2 - \mathrm{Re} (f_A^{(2)}|g_A^{(1)})_{-\frac{1-{\mu}^2}{2}}^2\,. $$
   % It follows from H\"older's inequality (see, e.g.~\cite{}) that $\mathrm{Re} (f|g_I^{(1)}-g_I^{(2)})_{s}^2 \leq ||f||_{s+s'}^2||g_I^{(1)}-g_I^{(2)}||_{s-s'}^2$ for $s'\in [-1,1]$. Applying this inequality with $s=-(1- \bar{\mu}^2)/2$ we find that $\mathfrak{b}_I \geq 0$.  Finally, in case 3) we find that  $\nu_I^2 -(\nu_{I}^{\mathrm{Mink}})^2= \mathfrak{c}_I \mathfrak{R}_H^{2-2\bar{\mu}^2}(1 + \mathcal{O}(\bar{\mu}^2)) + \mathcal{O}(\mathfrak{R}_H^{1-\bar{\mu}^2})$, with $$ \mathfrak{c}_I = ||f_I^{(2)}||^2_{-\frac{3}{2}-\bar{\mu}^2} ||g_I^{(1)}||_{\frac{1}{2}}^2 - \mathrm{Re} (f_I^{(2)}|g_I^{(1)})_{-\frac{1-\bar{\mu}^2}{2}}^2\,. $$
 % 
  Following the same argument as in case (2), we find  $\mathfrak{c}_I \geq 0$.  %In cases (2) and (3), the inequalities can be saturated, i.e., $\mathfrak{b}_I =0$ and $\mathfrak{c}_I=0$. 
  Saturation of these inequality can occur for special functions; the proof is modified in that case and will be given in~\cite{longpaper}.  
  %Saturation occurs when the functions $f^{(1)}_I$ and $f^{(2)}_I$ {\color{red} are insensitive to curvature at large scales }{\color{blue} IA: I don't think we provide enough info for this comment to be understood by readers. In that case, we may want not to mention when saturation happens; or give more details}.  For such functions, one can show that the massless limit is well-defined and $\nu_{I}^2 - (\nu_I^{\mathrm{Mink}})^2 =0$ in this limit.
   %{\color{red} We still need to comment on what happens when the inequalities in cases 2) and 3) are saturated, i.e. when $\mathfrak{b}_I =0$ and $\mathfrak{c}_I=0$. The equality in the previous expressions can only be attained if the functions $f_I$ in case 2) and $f^{(2)}_I$ in case 3) are unable to feel the effect of curvature at large scales. Moreover, for this family of functions, one can show that the massless limit is well-defined and $\nu_{I}^2 - (\nu_I^{\mathrm{Mink}})^2 =0$ in this limit. This concludes the proof. 
   % The inequality  $\mathfrak{b}_I \geq 0$ is saturated when $|\tilde{f}(\vec{  q})||\vec{q}|^{s'} = |\tilde{f}(\vec{  q})||\vec{q}|^{-s'}$ almost everywhere. 
   \end{proof}

$S(\nu_A)$ quantifies the entanglement between mode A and the rest of the field degrees of freedom. Hence, such entanglement is larger in the Bunch-Davies vacuum than in Minkowski vacuum. This result aligns with previous calculations of the entropy of a region \cite{maldacena_entanglement_2013}.  \\

{\bf Correlations.} Mutual information
%, being invariant under system-local symplectic transformations, 
provides an invariant way of quantifying the correlation between two single-mode subsystems. In terms of the entropy of each subsystem, it is  
\begin{equation}
  \mathcal{I}(A,B) =   S(\nu_A) +   S(\nu_B) -   S(\nu_+) - S(\nu_-) \,.
\end{equation}
Fig.~\ref{MI} shows an illustrative example.
\begin{prop}
    The mutual information between two single-mode subsystems of a scalar field in the Bunch-Davies vacuum is greater than in the Minkowski vacuum of flat spacetime when the supports of A and B and their separation are larger than the Hubble radius. 
    %
    %(1) When the supports of  modes $A$ and $B$,  are small compared to the Hubble radius ($R_I H \ll $, I=A,B), but the distance between them is large  ($\Delta x_{AB}  H\gg 1$); and (2)  When the supports of A and B are are large ($R_I H \gg 1$), regardless of their separation. {\color{blue} Fro consistency with other sections, should we restrict to case (2) and just  mention }
\end{prop}

\begin{comment}
\begin{prop}
    The mutual information of the single-mode subsystems $(\hat{O}_A^{(1)},\hat{O}_A^{(2)})$ and  $(\hat{O}_B^{(1)},\hat{O}_B^{(2)})$ of a scalar field in the Bunch-Davies vacuum of de Sitter spacetime is larger than in the mutual information of corresponding modes of a massless scalar field in the vacuum of Minkowski spacetime in the regimes when the curvature effects dominate, i.e., when 1) the size of the supports of the modes $A$ and $B$, given by $R_I:= (3/(4\pi))^{1/3}V_{I}^{1/3}$ with $V_I$  the volume of the support of mode $I = \{A,B\}$, are small compared to the Hubble radius ($R_A \ll H^{-1}$ and $R_B \ll H^{-1}$), but the distance between the regions of support of the modes, $\Delta x_{AB}$, is large compared to the same scale $\Delta x_{AB}\gg H^{-1}$;  and 2) when the size of the supports of the two modes are much larger than the Hubble radius, i.e., when $R_A \gg H^{-1}$ and $R_B \gg H^{-1}$. 
\end{prop}
\end{comment}
\begin{proof}
We showed above that, in the regime $RH \gg1$, $\nu_I^{2} \gg 1$, $I=A,B$. Similar arguments imply $\nu_{\pm}^{2} \gg 1$. Using this, the mutual information can be approximated as: 
\begin{equation}
     \mathcal{I}(A,B) \underset{RH\gg 1}{\sim} \frac{1}{2} \log_2\left(\frac{\nu_A^2 \nu_B^2}{\nu_+^2 \nu_-^2}\right) \,.
\end{equation}
The leading order dependence in $RH$ cancels out in the argument of the logarithm, making $\mathcal{I}(A,B)\sim \mathcal{O}((RH)^0)$, i.e., independent of $R H$, in the limit $RH \gg 1$. 

On the other hand, for large separations %the regions can be treated as point-like
%if the the two modes are far away relative to the size of their support, we can talk about the distance between them. 
%In this regime, the regions can be treated as point-like, and 
the mutual information's dependence on distance is governed by the asymptotic behavior of the correlations in Eqs.~\eqref{phiphi}-\eqref{phipi}. Employing standard tools of asymptotic harmonic analysis, we find that the large-separation behavior is dominated by the Sobolev product  $\mathrm{Re}(f_{A}^{(i)}|f_B^{(j)})_{-\frac{3}{2} + {\mu}^2}$, 
%contained in the field-field correlations, 
which decays  with the separation between modes as $\sim(\Delta x_{AB})^{-2{\mu}^2}$ %{\color{green} \st{$(\Delta x_{AB}/H)^{-2\bar{\mu}^2}$} In any case, this should be rewritten as $(\Delta x_{AB} H)^{-2\bar{\mu}^2}$. Nevertheless, I would say that it would be more correct to write $(\Delta x_{AB} H/(R H))^{-2\bar{\mu}^2}$ }
, and produces  $\mathcal{I}(A,B)\sim (\Delta x_{AB})^{-4{\mu}^2}$ ---which is almost scale invariant when ${\mu}\ll 1$.
In contrast, in the Minkowski limit, the term $(f_{A}^{(i)}|f_B^{(j)})_{-\frac{1}{2}} \sim (\Delta x_{AB})^{-2}$ dominates, producing the well-known result $\mathcal{I}^{\rm Mink}(A,B)\sim (\Delta x_{AB})^{-4}$. Hence, at large separations the almost-scale invariant term in the Bunch-Davies vacuum dominates. 
\end{proof}

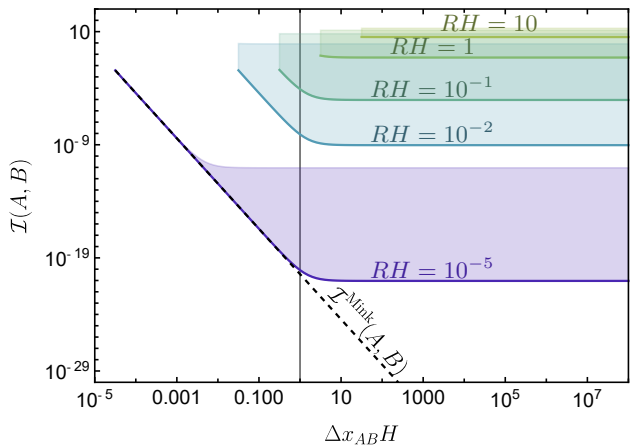
\begin{figure}
    \centering
    \begin{tikzpicture}
        \node (MI) at  (current page.center) {\includegraphics[width=0.46\textwidth]{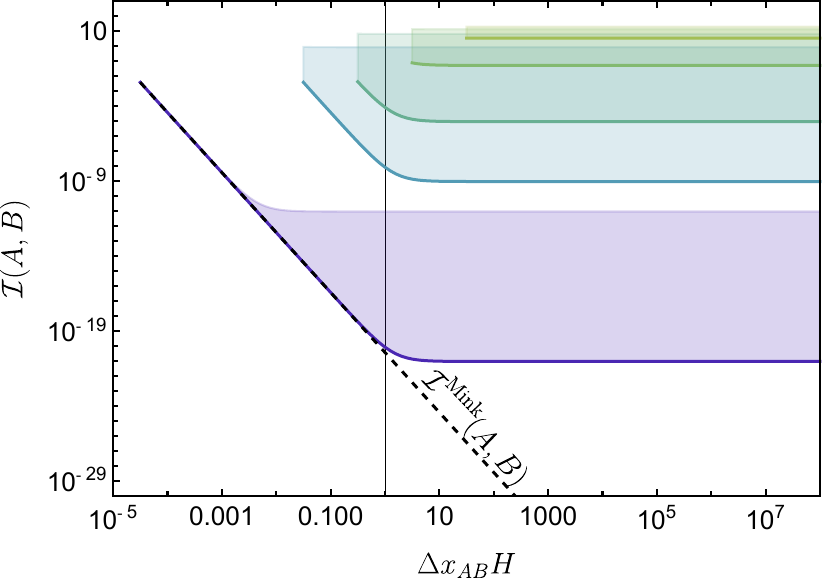}};
        \node[text={rgb:red,74;green,39;blue,178}] at ([xshift=1.5cm,yshift=-0.55cm]MI.center) {$RH = 10^{-5}$};
        \node[text={rgb:red,83;green,155;blue,181}] at ([xshift=1.5cm,yshift=1.25cm]MI.center) {$RH = 10^{-2}$};
        \node[text={rgb:red,104;green,175;blue,147}] at ([xshift=1.5cm,yshift=1.85cm]MI.center) {$RH = 10^{-1}$};
        \node[text={rgb:red,131;green,186;blue,112}] at ([xshift=1.5cm,yshift=2.4cm]MI.center) {$RH = 1$};
         \node[text={rgb:red,163;green,190;blue,86}] at ([xshift=2.25cm,yshift=2.7cm]MI.center) {$RH = 10$};   %\node[anchor=east] at ([xshift=4.5cm,yshift=3cm]MI.center) {$RH = 10^6$};
    \end{tikzpicture}
    \caption{Mutual information of two non-overlapping single-mode subsystems versus their separation (in units $H$). The chosen modes are spherically symmetric, and defined by $\gamma_I^{(1)}=(f_I(\vec x),0)$, $\gamma_I^{(2)}=(0,f_I(\vec x))$, with  $f_I(\vec x)= N\  \left(1 -\frac{|\vec{x} - \vec{x}_I|^2}{R^2}\right)\,  \Theta\left( 1 - \frac{|\vec{x} - \vec{x}_I|}{R}\right)$, $I=A,B$. $x_I$ denotes the center of each mode, $R$ the radius of their support, and $N$ is a normalization constant. This figure is obtained numerically (i.e., without analytical approximations). It shows that: (i) When $RH\ll1$ (sub-Hubble support) and $\Delta x_{AB}H\ll 1$ (sub-Hubble separation), $I(A,B)$ in the Bunch-Davies approaches the Minkowski value. %; in particular, it falls of as $\Delta x_{AB}^{-4}$. 
    (ii) For $RH\ll 1$ but $\Delta x_{AB}H\gg 1$  (sub-Hubble support but super-Hubble separation), $\mathcal{I}(A,B)$ becomes almost scale-invariant. %{\color{green} \st{(iii) When $RH\gg 1$, $I(A,B)$ becomes independent of $RH$}  PRM: This is not really visible anymore in this version of the plot. Should we remove it, of add more lines in a different color? } {\color{blue} Let's removed it then; we already say it in the text.}.
    The shaded region shows the variation of the mutual information when $\mu^2$ changes from  $\mu^2 = 10^{-10}$ (bottom) to  $\mu^2 = 10^{-15}$ (top). }
    \label{MI}
\end{figure}

%A powerful tool for this purpose is mutual information, which encapsulates the total amount of correlations between two subsystems. In this section, we will define mutual information and explore how it captures the change in correlations due to the curvature of spacetime, focusing on the localized observables in question.

% Correlations are a correlations between the pair of modes $(\hat{O}_A^{(1)},\hat{O}_A^{(2)})$ and $(\hat{O}_B^{(1)},\hat{O}_B^{(2)})$

{\bf Entanglement.} 
Logarithmic negativity (LN) %quantifies entanglement 
 is a measure of entanglement \cite{vidal02,plenio05,werner01bound} %and is 
applicable to pure and mixed states (unlike entropy). 
%If  $\hat \rho_{AB}$ denotes the density matrix of a two-mode system, % The Logarithmic Negativity with respect to the partition $A|B$ is defined as
%\be {\rm LN}(\hat \rho_{AB})=\log_2 ||\hat \rho^{\top_B}||_1\, \label{eq:defLN},\ee
%Here, $\hat \rho$ represents the density matrix of the system, 
%where $\hat \rho^{\top_B}$ represents the partial transpose of $\hat \rho_{AB}$ with respect to subsystem  B, and $||\cdot||_1$  the trace norm. %A non-zero LN value indicates a violation of the Positivity of the Partial Transpose criterion for quantum states \cite{peres96}. 
For Gaussian states and two-mode systems, LN can be computed from the invariant $\tilde \nu_-$ as $ {\rm LN}(\hat \rho_{AB})={\rm max}[0,-\log_2 \tilde \nu_-].$ 
LN is non-zero if and only if the state is entangled; a higher LN value corresponds to more entanglement.
\begin{prop}
    The entanglement between any two non-overlapping and compactly supported single-mode subsystems of a scalar field with mass $m \ll H$ prepared in the Bunch-Davies vacuum is, when the supports of each mode are larger than the Hubble radius ($R_I H \gg 1$), no bigger than it would be in the Minkowski vacuum, regardless of their separation. 
    %
    %(1) When the supports of  modes $A$ and $B$,  are small compared to the Hubble radius ($R_I H \ll $, I=A,B), but the distance between them is large  ($\Delta x_{AB}  H\gg 1$); and (2)  When the supports of A and B are are large ($R_I H \gg 1$), regardless of their separation. {\color{blue} Fro consistency with other sections, should we restrict to case (2) and just  mention }
\end{prop}
\begin{proof}
    Because LN decreases monotonically  when $\tilde \nu_-$ increases, it suffices to prove that $\tilde \nu_-$ in the Bunch-Davies vacuum is larger than or equal to $\tilde \nu_-$ in the Minkowski vacuum when $m/H\ll 1$ and $RH\gg 1$. In this regime, we find  
\be \tilde \nu_-^2=(\tilde \nu^{\rm Mink}_-)^2+(RH)^{4-4{\mu}^2} \tilde{\mathcal{F}}_-\, (1 + \mathcal{O}({\mu}^2)) + \mathcal{O}((RH)^3).
\ee
Using again the Cauchy–Schwarz inequality for Sobolev norms and  Fischer's inequality, we find that $\tilde{\mathcal{F}}_-$, which depends on the form of the smearing functions and the separation between the two modes, but not on   $RH$, is non-negative. 
This proof assumes $(f_I|f_J)_{-\frac{3}{2}}\neq 0$, but can be extended to  choices of smearing functions for which this is not satisfied \cite{longpaper}. 
%
%{\color{green} You are right,  non-negative is more rigorous. Also, the particular form in Eq. (12) assumes that the modes we chose are not restricted in the sense of the classification in Prop. 1. Maybe we should mention that here we assumed that we are in a similar situation as case 1) above (all the smearing functions are non-zero and the degrees of freedom in A and B are different). One also needs to prove this statement for all possible special cases (some of the functions coincide, are 0, or are designed in such a way that they do not see the long-distance correlations). We omit the proof for these special cases as it is similar in spirit to the proof outlined here. Note that in some concrete cases, the proof can be strengthened to every regime of $RH$...}.
%Thus, for large regions $\mathfrak{R}_H \gg 1$, the smallest symplectic eigenvalue $\tilde{\nu}_-^2$ in de Sitter spacetime exceeds its counterpart for corresponding modes in Minkowski spacetime, indicating a suppression of entanglement. 
%\end{proof}
\end{proof}

We conclude that local super-Hubble  modes are less \textit{entangled} in the Bunch-Davies vacuum than in Minkowski despite being more \textit{correlated}. See Fig.~\ref{fig:entanglement} for an example.
\\

%it has been shown above using Mutual Information. 

%This is true even  though they are more correlated, as it has been shown above using Mutual Information. 

%We conclude that, although the analysts of Mutual Information shows that local modes are more correlated in the Bunch-Davies vacuum, they are less entangled. Hence, the correlations encoded in the Mutual Information do not originate from quantum entanglement.

%We conclude that, although the analysts of Mutual Information shows that local modes are more correlated in the Bunch-Davies vacuum, they are less entangled. Hence, the correlations encoded in the Mutual Information do not originate from quantum entanglement.\\

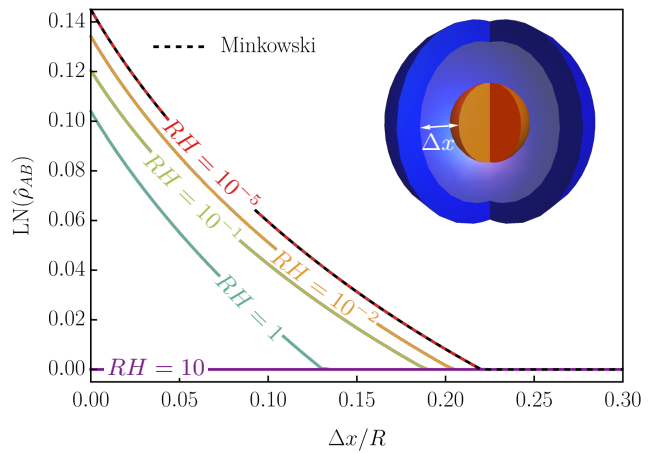
\begin{figure}
    \centering
    \begin{tikzpicture}
       \node (C) at (current page.center){ \includegraphics[width=0.47\textwidth]{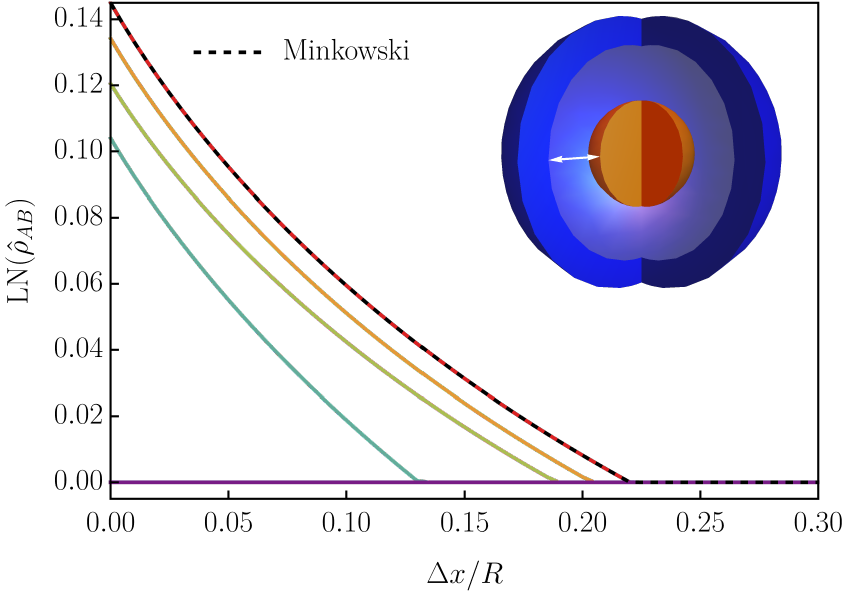}};
       \node at ([xshift=1.5cm,yshift=1.2cm]C.center) {\color{white}$\Delta x$}; 
       \node[fill=white,rotate=-45,inner sep=0pt] at ([xshift=-1.5cm,yshift=0.9cm]C.center) {\color{mrainbow1} $RH = 10^{-5}$ };
       \node[fill=white,rotate=-36,inner sep=0pt] at ([xshift=.cm,yshift=-0.75cm]C.center) {\color{mrainbow2} $RH = 10^{-2}$ };
       \node[fill=white,rotate=-44,inner sep=0pt] at ([xshift=-1.75cm,yshift=0.45cm]C.center) {\color{mrainbow3} $RH = 10^{-1}$ };
        \node[fill=white,rotate=-45,inner sep=0pt] at ([xshift=-1.cm,yshift=-1.cm]C.center) {\color{mrainbow4} $RH = 1$ };
         \node[fill=white,inner sep=1pt] at ([xshift=-2.25cm,yshift=-1.8cm]C.center) {\color{mrainbow6} $RH = 10$ };
    \end{tikzpicture}
    \caption{LN between two single-modes, compactly supported within a sphere of radius $R$ and a shell around it. The mode in the ball is  defined by $\gamma_A^{(1)} = (0,f_A(\vec{x}))$ and $\gamma_A^{(2)} = (-f_A(\vec{x}),0)$, where $f_A(\vec{x}) = N \left(1-\frac{|\vec{x}|^2}{R^2} \right) \Theta\left(1-\frac{|\vec{x}|}{R} \right)$. The mode in the shell is defined by $\gamma_S^{(1)} = (0,f_S(\vec{x}))$ and $\gamma_S^{(2)} = (-f_S(\vec{x}),0)$, where   $f_{S}(\vec{x}) = \Big(|\vec{x}|-(R_S - d)\Big) \Big((R_S+d)-|\vec{x}| \Big) \Theta\Big(|\vec{x}|-(R_S-d)\Big)\Theta\Big((R_S+d) - |\vec{x}|\Big)\,,$ with $R_S \pm d$ are the outer/inner radii of the shell and $\mu=10^{-2}$. %{\color{green}PRM: $\mu = 10^{-2}$. This means that $m/H = \sqrt{9/4 - (\mu^2 - 3/2)^2} \approx 0.0173 $ }
This figure shows: (i) LN falls of exponentially with the radial distance $\Delta x$ between shell and ball modes. (ii) LN decreases when $H$ increases, illustrating the content of Proposition 3---LN agrees with the Minkowski result when $RH\to 0$ \cite{Agullo:2023fnp}, and  vanishes for $RH$ larger than a threshold value.}
    \label{fig:entanglement}
\end{figure}

{\bf Entanglement between a mode and a region.}
We extend the previous analysis by evaluating the entanglement between a single mode and a set of modes supported within a region. We focus here on a specific example. Nevertheless,  within this limitation, the analysis serves to extend the previous discussion beyond pairs of modes.

The set up is the following. Subsystem $A$ consists of a single mode supported on a spherical shell; we use the same mode as in Fig.~\ref{fig:entanglement}. Subsystem $B$ is made of modes supported within a sphere concentric with the shell and of smaller radius. To construct a basis for these modes, we start with the polynomial functions $f^{(\delta)}(\vec{x}) = \left(1 - \frac{|\vec{x}|^2}{R^2}\right)^{\delta}\, \Theta\left(1 - \frac{|\vec{x}|}{R}\right)$, where $\delta \in \mathbb{Z}_+$. We then use a symplectic version of the Gram-Schmidt orthogonalization algorithm to generate a commuting set of modes. By restricting $\delta \leq n_B$, we introduce an ultraviolet cut-off within subsystem $B$.

Fig.~\ref{fig:shellballmany} shows the logarithmic negativity versus $n_B$.\\

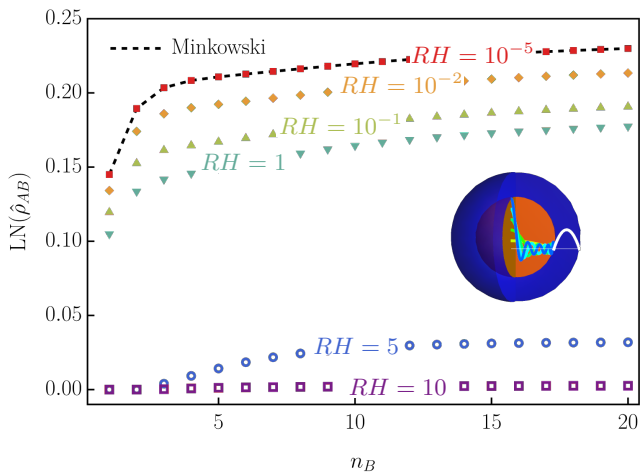
\begin{figure}
\begin{tikzpicture}
    \node (c) at (current page.center) {\includegraphics[width=0.47\textwidth]{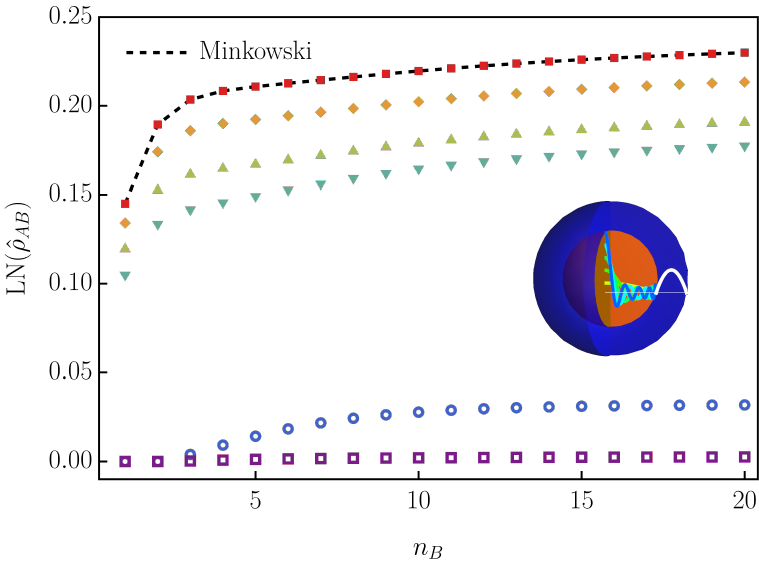}};
    \node[fill=white,inner sep=0pt] at ([xshift=2cm,yshift=2.5cm]c.center) {{\color{mrainbow1}$RH = 10^{-5}$}};
    \node[fill=white,inner sep=0pt] at ([xshift=1.05cm,yshift=2.1cm]c.center) {\color{mrainbow2}$RH = 10^{-2}$};
    \node[fill=white,inner sep=1pt] at ([xshift=0.25cm,yshift=1.55cm]c.center) {\color{mrainbow3}$RH = 10^{-1}$};
    \node[fill=white,inner sep=1pt] at ([xshift=-1.05cm,yshift=1.05cm]c.center) {\color{mrainbow4}$RH = 1$};
    \node[fill=white,inner sep=0pt] at ([xshift=0.45cm,yshift=-1.4cm]c.center) {\color{mrainbow5}$RH = 5$};
    \node[fill=white,inner sep=1pt] at ([xshift=1cm,yshift=-1.95cm]c.center) {\color{mrainbow6}$RH = 10$};
\end{tikzpicture}
\caption{LN between a mode supported in a spherical shell and a set of $n_B$ independent modes supported within a sphere (we use $\mu = 10^{-2}$). To maximize entanglement, the radius of the sphere is chosen to coincide with the inner radius of the shell. This figure shows: (i) LN changes with the cut-off $n_B$. (ii) LN monotonically decreases as $H$ increases, showing that de Sitter's curvature reduces the entanglement between $A$ and $B$. }
%This plot corresponds to $m = 10^{-2}\, H$.
%{\color{blue} We have checked that the conclusions drawn from this figure remain valid for larger $m$, even when $m > H$.}}
\label{fig:shellballmany}
\end{figure}

%.

{\bf Where is the entanglement?} We have shown that increasing $H$ decreases the entanglement between pairs of super-Hubble, non-overlapping modes, while the entropy of each mode increases. Since this entropy measures the degree of entanglement between the mode and the rest of the field modes, larger entropy indicates that the theory contains more entanglement for larger $H$.

These two results seem in tension. To understand why they are not, we consider the \emph{partner}  of a given mode $A$, denoted as $A_P$ onwards. It is defined as the single-mode subsystem that purifies $A$, i.e., the mode for which the reduced state $\hat{\rho}_{AA_P}$ is pure. If the field is prepared in a pure state, the partner of $A$ exists and is unique. It can be computed using the tools in \cite{Botero_2003,hotta2015partner,Trevison_2019} (see also \cite{partnerformula}).

\begin{prop}
The partner $A_P$ of a compactly supported single-mode $A$ is, when the field is prepared in the Bunch-Davies vacuum, not compactly supported. For typical modes $A$, the smearing functions $f^{(i)}_{Ap}(\vec x)$
 defining $A_P$ fall off at least as $r^{-2{\mu}^2}$ as $r \to \infty$ when $\mu \ll 1$. In contrast, in the Minkowski vacuum, %{\color{green} `the slowest possible decaying' smearing function for the partner} 
 they fall off significantly faster, as $r^{-2}$ for $m=0$, and as $e^{-m r}$ for $m \neq 0$.
\end{prop}

This proposition follows from the almost-scale invariant form of the two-point correlation functions of field operators in the Bunch-Davies vacuum, and the fact that the smearing functions defining the partner mode are obtained by integrating the smearing functions defining $A$ against the two-point correlation functions. For the proof, see \cite{longpaper}.
%{\color{red} (Mention special functions).} 
%
%{\color{green} [The decay we mention in the proposition above can be altered (can be faster) if one chooses special functions (designed to `not feel' long-distance correlations), or by fine-tuning the modes (e.g. by choosing some functions to vanish).% The important note is that the slowest possible decay in the proposition above is a universal property of the spacetime and the state of the quantum field (in the sense that it is fully determined by the behavior of the two-point function). }

The partner mode encodes all correlations, classical and quantum, of mode $A$ with the rest of the field degrees of freedom. Because the reduced state $\hat \rho_{AA_P}$ is pure, we can use the von Neumann entropy of $A$ to measure the quantum correlations. 
%The partner mode encodes all correlations, classical and quantum, mode $A$ has with the rest of the field degrees of freedom. Because Bunch-Davies is a pure state, these correlations are quantified by the von Neumann entropy of $A$. 
We found in Prop.~1 that this entropy is larger in the Bunch-Davies vacuum than in  Minkowski, meaning that $A$ is more correlated and entangled with its partner.

Proposition 4 informs us about the spatial distribution of the partner mode, or equivalently,  about the distribution of the correlations and entanglement with mode $A$ that are contained in the Bunch-Davies vacuum: they are almost scale-invariant, hence spread over much longer distances than they would in Minkowski spacetime. Local observers do not have access to the entirety of the  partner mode. For them, these long-distance correlations manifest as a larger entropy for local modes. Physically, this entropy corresponds to local thermal noise, which decreases the entanglement between pairs of compactly supported modes. %{\color{red} IA: please, read carefully the last two paras.}
\\

%the partner mode is not accessible to local observers, as the correlations are almost scale-invariant, hence distributed over much longer distances than in Minkowski spacetime. For local observers, the correlations between $A$ and $A_P$ manifest as a larger entropy for local modes. Physically, this entropy corresponds to local thermal noise, which decreases the entanglement between compactly supported modes.

%\\

{\bf Discussion.}
We emphasized the importance of focusing on localized field modes to study the entanglement generated by inflation and its detectability. The formalism we used is local, free of ultraviolet divergences, and focused on system-local symplectic invariants.

Previous findings in the literature appear contradictory at first glance. On the one hand, two local particle detectors in de Sitter spacetime harvest less entanglement than in Minkowski \cite{Perche:2023nde,Steeg_2009}, suggesting that the entanglement content of the Bunch-Davies vacuum is smaller than its flat-spacetime counterpart. On the other hand, calculations of the entropy of sub-regions in de Sitter spacetime indicate the opposite \cite{maldacena_entanglement_2013}. We have found that, not only are these two results not contradictory, but one is responsible for the other: the large degree of correlations and entanglement that local modes exhibit with their partner modes, which are supported at arbitrarily large separations, is responsible for the local thermal properties of de Sitter. This local thermal entropy, in turn, acts as a natural decoherence factor for local modes.

Beyond the implications for early-universe cosmology, the analysis in this article has ramifications for formal aspects of de Sitter spacetime and, more generally, for the understanding of the interplay between curvature and entanglement in quantum field theory.

\begin{acknowledgements}
We thank J. Martin, W. van Suijlekom, and V. Vennin for discussions.  This research has been partially supported by the Blaumann Foundation.  
IA and BB both acknowledge the support of the Visiting Fellow program of Perimeter Institute for Theoretical Physics, and PRM thanks the Perimeter Institute for Theoretical Physics for their support as well as LSU for her short-term visit.  IA is additionally supported by the Hearne Institute for Theoretical Physics, by the NSF grants PHY-2409402 and PHY-2110273, by the RCS program of Louisiana Boards of Regents through the grant LEQSF(2023-25)-RD-A-04.
Research at Perimeter Institute is supported in part by the Government of Canada through the Department of Innovation, Science and Industry Canada and by the Province of Ontario through the Ministry of Colleges and Universities. %\textcolor{purple}{EMM, TRP, JPG, and BSLT thank the WINQ program for hosting them in the WINQ program on Complex and Quantum Systems, where they became aware of a series of works relevant to this manuscript.}  %Perimeter Institute and the University of Waterloo are situated on the Haldimand Tract, land that was promised to the Haudenosaunee of the Six Nations of the Grand River, and is within the territory of the Neutral, Anishinaabe, and Haudenosaunee people.
\end{acknowledgements}

%Our results reconciled to 
%Does the Bunch-Davies state contain more entanglement than Minkowski vacuum? We have analyzed this question from the view point of finite dimensional subsystem, for which the analysses is free of UV divergences. Our answer is ``yes and no''. Each individual mode compactly supported is 
%which are free of the UV divergences 

%is the origin of the Gibbons-Hawking thermal properties of the Bunch-Davies vacuum,

%=================\\
\newpage

%\vspace{3cm}

%\bibliographystyle{apsrev4-1}
\bibliography{dSbib.bib} 

%apsrev4-2.bst 2019-01-14 (MD) hand-edited version of apsrev4-1.bst
%Control: key (0)
%Control: author (8) initials jnrlst
%Control: editor formatted (1) identically to author
%Control: production of article title (0) allowed
%Control: page (0) single
%Control: year (1) truncated
%Control: production of eprint (0) enabled
\begin{thebibliography}{52}%
\makeatletter
\providecommand \@ifxundefined [1]{%
 \@ifx{#1\undefined}
}%
\providecommand \@ifnum [1]{%
 \ifnum #1\expandafter \@firstoftwo
 \else \expandafter \@secondoftwo
 \fi
}%
\providecommand \@ifx [1]{%
 \ifx #1\expandafter \@firstoftwo
 \else \expandafter \@secondoftwo
 \fi
}%
\providecommand \natexlab [1]{#1}%
\providecommand \enquote  [1]{``#1''}%
\providecommand \bibnamefont  [1]{#1}%
\providecommand \bibfnamefont [1]{#1}%
\providecommand \citenamefont [1]{#1}%
\providecommand \href@noop [0]{\@secondoftwo}%
\providecommand \href [0]{\begingroup \@sanitize@url \@href}%
\providecommand \@href[1]{\@@startlink{#1}\@@href}%
\providecommand \@@href[1]{\endgroup#1\@@endlink}%
\providecommand \@sanitize@url [0]{\catcode `\\12\catcode `\$12\catcode
  `\&12\catcode `\#12\catcode `\^12\catcode `\_12\catcode `\%12\relax}%
\providecommand \@@startlink[1]{}%
\providecommand \@@endlink[0]{}%
\providecommand \url  [0]{\begingroup\@sanitize@url \@url }%
\providecommand \@url [1]{\endgroup\@href {#1}{\urlprefix }}%
\providecommand \urlprefix  [0]{URL }%
\providecommand \Eprint [0]{\href }%
\providecommand \doibase [0]{https://doi.org/}%
\providecommand \selectlanguage [0]{\@gobble}%
\providecommand \bibinfo  [0]{\@secondoftwo}%
\providecommand \bibfield  [0]{\@secondoftwo}%
\providecommand \translation [1]{[#1]}%
\providecommand \BibitemOpen [0]{}%
\providecommand \bibitemStop [0]{}%
\providecommand \bibitemNoStop [0]{.\EOS\space}%
\providecommand \EOS [0]{\spacefactor3000\relax}%
\providecommand \BibitemShut  [1]{\csname bibitem#1\endcsname}%
\let\auto@bib@innerbib\@empty
%</preamble>
\bibitem [{\citenamefont {Grishchuk}\ and\ \citenamefont
  {Sidorov}(1990)}]{Grishchuk:1990bj}%
  \BibitemOpen
  \bibfield  {author} {\bibinfo {author} {\bibfnamefont {L.~P.}\ \bibnamefont
  {Grishchuk}}\ and\ \bibinfo {author} {\bibfnamefont {Y.~V.}\ \bibnamefont
  {Sidorov}},\ }\bibfield  {title} {\bibinfo {title} {{Squeezed quantum states
  of relic gravitons and primordial density fluctuations}},\ }\href
  {https://doi.org/10.1103/PhysRevD.42.3413} {\bibfield  {journal} {\bibinfo
  {journal} {Phys. Rev. D}\ }\textbf {\bibinfo {volume} {42}},\ \bibinfo
  {pages} {3413} (\bibinfo {year} {1990})}\BibitemShut {NoStop}%
\bibitem [{\citenamefont {Albrecht}\ \emph {et~al.}(1994)\citenamefont
  {Albrecht}, \citenamefont {Ferreira}, \citenamefont {Joyce},\ and\
  \citenamefont {Prokopec}}]{Albrecht:1992kf}%
  \BibitemOpen
  \bibfield  {author} {\bibinfo {author} {\bibfnamefont {A.}~\bibnamefont
  {Albrecht}}, \bibinfo {author} {\bibfnamefont {P.}~\bibnamefont {Ferreira}},
  \bibinfo {author} {\bibfnamefont {M.}~\bibnamefont {Joyce}},\ and\ \bibinfo
  {author} {\bibfnamefont {T.}~\bibnamefont {Prokopec}},\ }\bibfield  {title}
  {\bibinfo {title} {{Inflation and squeezed quantum states}},\ }\href
  {https://doi.org/10.1103/PhysRevD.50.4807} {\bibfield  {journal} {\bibinfo
  {journal} {Phys. Rev. D}\ }\textbf {\bibinfo {volume} {50}},\ \bibinfo
  {pages} {4807} (\bibinfo {year} {1994})},\ \Eprint
  {https://arxiv.org/abs/astro-ph/9303001} {arXiv:astro-ph/9303001}
  \BibitemShut {NoStop}%
\bibitem [{\citenamefont {Lesgourgues}\ \emph {et~al.}(1997)\citenamefont
  {Lesgourgues}, \citenamefont {Polarski},\ and\ \citenamefont
  {Starobinsky}}]{Lesgourgues:1996jc}%
  \BibitemOpen
  \bibfield  {author} {\bibinfo {author} {\bibfnamefont {J.}~\bibnamefont
  {Lesgourgues}}, \bibinfo {author} {\bibfnamefont {D.}~\bibnamefont
  {Polarski}},\ and\ \bibinfo {author} {\bibfnamefont {A.~A.}\ \bibnamefont
  {Starobinsky}},\ }\bibfield  {title} {\bibinfo {title} {{Quantum to classical
  transition of cosmological perturbations for nonvacuum initial states}},\
  }\href {https://doi.org/10.1016/S0550-3213(97)00224-1} {\bibfield  {journal}
  {\bibinfo  {journal} {Nucl. Phys. B}\ }\textbf {\bibinfo {volume} {497}},\
  \bibinfo {pages} {479} (\bibinfo {year} {1997})},\ \Eprint
  {https://arxiv.org/abs/gr-qc/9611019} {arXiv:gr-qc/9611019} \BibitemShut
  {NoStop}%
\bibitem [{\citenamefont {Kiefer}\ \emph
  {et~al.}(1998{\natexlab{a}})\citenamefont {Kiefer}, \citenamefont
  {Lesgourgues}, \citenamefont {Polarski},\ and\ \citenamefont
  {Starobinsky}}]{Kiefer:1998pb}%
  \BibitemOpen
  \bibfield  {author} {\bibinfo {author} {\bibfnamefont {C.}~\bibnamefont
  {Kiefer}}, \bibinfo {author} {\bibfnamefont {J.}~\bibnamefont {Lesgourgues}},
  \bibinfo {author} {\bibfnamefont {D.}~\bibnamefont {Polarski}},\ and\
  \bibinfo {author} {\bibfnamefont {A.~A.}\ \bibnamefont {Starobinsky}},\
  }\bibfield  {title} {\bibinfo {title} {{The Coherence of primordial
  fluctuations produced during inflation}},\ }\href
  {https://doi.org/10.1088/0264-9381/15/10/002} {\bibfield  {journal} {\bibinfo
   {journal} {Class. Quant. Grav.}\ }\textbf {\bibinfo {volume} {15}},\
  \bibinfo {pages} {L67} (\bibinfo {year} {1998}{\natexlab{a}})},\ \Eprint
  {https://arxiv.org/abs/gr-qc/9806066} {arXiv:gr-qc/9806066} \BibitemShut
  {NoStop}%
\bibitem [{\citenamefont {Kiefer}\ \emph
  {et~al.}(1998{\natexlab{b}})\citenamefont {Kiefer}, \citenamefont
  {Polarski},\ and\ \citenamefont {Starobinsky}}]{Kiefer:1998qe}%
  \BibitemOpen
  \bibfield  {author} {\bibinfo {author} {\bibfnamefont {C.}~\bibnamefont
  {Kiefer}}, \bibinfo {author} {\bibfnamefont {D.}~\bibnamefont {Polarski}},\
  and\ \bibinfo {author} {\bibfnamefont {A.~A.}\ \bibnamefont {Starobinsky}},\
  }\bibfield  {title} {\bibinfo {title} {{Quantum to classical transition for
  fluctuations in the early universe}},\ }\href
  {https://doi.org/10.1142/S0218271898000292} {\bibfield  {journal} {\bibinfo
  {journal} {Int. J. Mod. Phys. D}\ }\textbf {\bibinfo {volume} {7}},\ \bibinfo
  {pages} {455} (\bibinfo {year} {1998}{\natexlab{b}})},\ \Eprint
  {https://arxiv.org/abs/gr-qc/9802003} {arXiv:gr-qc/9802003} \BibitemShut
  {NoStop}%
\bibitem [{\citenamefont {Kiefer}\ and\ \citenamefont
  {Polarski}(2009)}]{Kiefer:2008ku}%
  \BibitemOpen
  \bibfield  {author} {\bibinfo {author} {\bibfnamefont {C.}~\bibnamefont
  {Kiefer}}\ and\ \bibinfo {author} {\bibfnamefont {D.}~\bibnamefont
  {Polarski}},\ }\bibfield  {title} {\bibinfo {title} {{Why do cosmological
  perturbations look classical to us?}},\ }\href
  {https://doi.org/10.1166/asl.2009.1023} {\bibfield  {journal} {\bibinfo
  {journal} {Adv. Sci. Lett.}\ }\textbf {\bibinfo {volume} {2}},\ \bibinfo
  {pages} {164} (\bibinfo {year} {2009})},\ \Eprint
  {https://arxiv.org/abs/0810.0087} {arXiv:0810.0087 [astro-ph]} \BibitemShut
  {NoStop}%
\bibitem [{\citenamefont {Polarski}\ and\ \citenamefont
  {Starobinsky}(1996)}]{Polarski:1995jg}%
  \BibitemOpen
  \bibfield  {author} {\bibinfo {author} {\bibfnamefont {D.}~\bibnamefont
  {Polarski}}\ and\ \bibinfo {author} {\bibfnamefont {A.~A.}\ \bibnamefont
  {Starobinsky}},\ }\bibfield  {title} {\bibinfo {title} {{Semiclassicality and
  decoherence of cosmological perturbations}},\ }\href
  {https://doi.org/10.1088/0264-9381/13/3/006} {\bibfield  {journal} {\bibinfo
  {journal} {Class. Quant. Grav.}\ }\textbf {\bibinfo {volume} {13}},\ \bibinfo
  {pages} {377} (\bibinfo {year} {1996})},\ \Eprint
  {https://arxiv.org/abs/gr-qc/9504030} {arXiv:gr-qc/9504030} \BibitemShut
  {NoStop}%
\bibitem [{\citenamefont {Martin}\ and\ \citenamefont
  {Vennin}(2016)}]{Martin:2015qta}%
  \BibitemOpen
  \bibfield  {author} {\bibinfo {author} {\bibfnamefont {J.}~\bibnamefont
  {Martin}}\ and\ \bibinfo {author} {\bibfnamefont {V.}~\bibnamefont
  {Vennin}},\ }\bibfield  {title} {\bibinfo {title} {{Quantum Discord of Cosmic
  Inflation: Can we Show that CMB Anisotropies are of Quantum-Mechanical
  Origin?}},\ }\href {https://doi.org/10.1103/PhysRevD.93.023505} {\bibfield
  {journal} {\bibinfo  {journal} {Phys. Rev. D}\ }\textbf {\bibinfo {volume}
  {93}},\ \bibinfo {pages} {023505} (\bibinfo {year} {2016})},\ \Eprint
  {https://arxiv.org/abs/1510.04038} {arXiv:1510.04038 [astro-ph.CO]}
  \BibitemShut {NoStop}%
\bibitem [{\citenamefont {Ashtekar}\ \emph {et~al.}(2020)\citenamefont
  {Ashtekar}, \citenamefont {Corichi},\ and\ \citenamefont {Kesavan}}]{ack}%
  \BibitemOpen
  \bibfield  {author} {\bibinfo {author} {\bibfnamefont {A.}~\bibnamefont
  {Ashtekar}}, \bibinfo {author} {\bibfnamefont {A.}~\bibnamefont {Corichi}},\
  and\ \bibinfo {author} {\bibfnamefont {A.}~\bibnamefont {Kesavan}},\
  }\bibfield  {title} {\bibinfo {title} {{Emergence of classical behavior in
  the early universe}},\ }\href {https://doi.org/10.1103/PhysRevD.102.023512}
  {\bibfield  {journal} {\bibinfo  {journal} {Phys. Rev. D}\ }\textbf {\bibinfo
  {volume} {102}},\ \bibinfo {pages} {023512} (\bibinfo {year} {2020})},\
  \Eprint {https://arxiv.org/abs/2004.10684} {arXiv:2004.10684 [gr-qc]}
  \BibitemShut {NoStop}%
\bibitem [{\citenamefont {Micheli}\ and\ \citenamefont
  {Peter}(2022)}]{Micheli:2022tld}%
  \BibitemOpen
  \bibfield  {author} {\bibinfo {author} {\bibfnamefont {A.}~\bibnamefont
  {Micheli}}\ and\ \bibinfo {author} {\bibfnamefont {P.}~\bibnamefont
  {Peter}},\ }\bibfield  {title} {\bibinfo {title} {{Quantum Cosmological
  Gravitational Waves?}},\ }\href@noop {} {\bibfield  {journal} {\bibinfo
  {journal} {arXiv e-prints}\ } (\bibinfo {year} {2022})},\ \Eprint
  {https://arxiv.org/abs/2211.00182} {arXiv:2211.00182 [gr-qc]} \BibitemShut
  {NoStop}%
\bibitem [{\citenamefont {Brahma}\ and\ \citenamefont
  {Seenivasan}(2023{\natexlab{a}})}]{Brahma:2023lqm}%
  \BibitemOpen
  \bibfield  {author} {\bibinfo {author} {\bibfnamefont {S.}~\bibnamefont
  {Brahma}}\ and\ \bibinfo {author} {\bibfnamefont {A.~N.}\ \bibnamefont
  {Seenivasan}},\ }\bibfield  {title} {\bibinfo {title} {{Probing the curvature
  of the cosmos from quantum entanglement due to gravity}},\ }\href@noop {}
  {\bibfield  {journal} {\bibinfo  {journal} {arXiv e-prints}\ } (\bibinfo
  {year} {2023}{\natexlab{a}})},\ \Eprint {https://arxiv.org/abs/2311.05483}
  {arXiv:2311.05483 [gr-qc]} \BibitemShut {NoStop}%
\bibitem [{\citenamefont {Brahma}\ and\ \citenamefont
  {Seenivasan}(2023{\natexlab{b}})}]{Brahma:2023uab}%
  \BibitemOpen
  \bibfield  {author} {\bibinfo {author} {\bibfnamefont {S.}~\bibnamefont
  {Brahma}}\ and\ \bibinfo {author} {\bibfnamefont {A.~N.}\ \bibnamefont
  {Seenivasan}},\ }\bibfield  {title} {\bibinfo {title} {{Gravity-induced
  entanglement as a probe of spacetime curvature}},\ }\href
  {https://doi.org/10.1142/S0218271823420208} {\bibfield  {journal} {\bibinfo
  {journal} {Int. J. Mod. Phys. D}\ }\textbf {\bibinfo {volume} {32}},\
  \bibinfo {pages} {2342020} (\bibinfo {year} {2023}{\natexlab{b}})},\ \Eprint
  {https://arxiv.org/abs/2310.17311} {arXiv:2310.17311 [gr-qc]} \BibitemShut
  {NoStop}%
\bibitem [{\citenamefont {Brahma}\ \emph {et~al.}(2023)\citenamefont {Brahma},
  \citenamefont {Calder\'on-Figueroa}, \citenamefont {Hassan},\ and\
  \citenamefont {Mi}}]{Brahma:2023hki}%
  \BibitemOpen
  \bibfield  {author} {\bibinfo {author} {\bibfnamefont {S.}~\bibnamefont
  {Brahma}}, \bibinfo {author} {\bibfnamefont {J.}~\bibnamefont
  {Calder\'on-Figueroa}}, \bibinfo {author} {\bibfnamefont {M.}~\bibnamefont
  {Hassan}},\ and\ \bibinfo {author} {\bibfnamefont {X.}~\bibnamefont {Mi}},\
  }\bibfield  {title} {\bibinfo {title} {{Momentum-space entanglement entropy
  in de Sitter spacetime}},\ }\href
  {https://doi.org/10.1103/PhysRevD.108.043522} {\bibfield  {journal} {\bibinfo
   {journal} {Phys. Rev. D}\ }\textbf {\bibinfo {volume} {108}},\ \bibinfo
  {pages} {043522} (\bibinfo {year} {2023})},\ \Eprint
  {https://arxiv.org/abs/2302.13894} {arXiv:2302.13894 [hep-th]} \BibitemShut
  {NoStop}%
\bibitem [{\citenamefont {Brahma}\ \emph {et~al.}(2024)\citenamefont {Brahma},
  \citenamefont {Calder\'on-Figueroa},\ and\ \citenamefont
  {Luo}}]{Brahma:2024yor}%
  \BibitemOpen
  \bibfield  {author} {\bibinfo {author} {\bibfnamefont {S.}~\bibnamefont
  {Brahma}}, \bibinfo {author} {\bibfnamefont {J.}~\bibnamefont
  {Calder\'on-Figueroa}},\ and\ \bibinfo {author} {\bibfnamefont
  {X.}~\bibnamefont {Luo}},\ }\bibfield  {title} {\bibinfo {title}
  {{Time-convolutionless cosmological master equations: Late-time resummations
  and decoherence for non-local kernels}},\ }\href@noop {} {\bibfield
  {journal} {\bibinfo  {journal} {arXiv e-prints}\ } (\bibinfo {year}
  {2024})},\ \Eprint {https://arxiv.org/abs/2407.12091} {arXiv:2407.12091
  [hep-th]} \BibitemShut {NoStop}%
\bibitem [{\citenamefont {Bhattacharyya}\ \emph {et~al.}(2024)\citenamefont
  {Bhattacharyya}, \citenamefont {Brahma}, \citenamefont {Haque}, \citenamefont
  {Lund},\ and\ \citenamefont {Paul}}]{Bhattacharyya:2024duw}%
  \BibitemOpen
  \bibfield  {author} {\bibinfo {author} {\bibfnamefont {A.}~\bibnamefont
  {Bhattacharyya}}, \bibinfo {author} {\bibfnamefont {S.}~\bibnamefont
  {Brahma}}, \bibinfo {author} {\bibfnamefont {S.~S.}\ \bibnamefont {Haque}},
  \bibinfo {author} {\bibfnamefont {J.~S.}\ \bibnamefont {Lund}},\ and\
  \bibinfo {author} {\bibfnamefont {A.}~\bibnamefont {Paul}},\ }\bibfield
  {title} {\bibinfo {title} {{The early universe as an open quantum system:
  complexity and decoherence}},\ }\href
  {https://doi.org/10.1007/JHEP05(2024)058} {\bibfield  {journal} {\bibinfo
  {journal} {JHEP}\ }\textbf {\bibinfo {volume} {05}},\ \bibinfo {pages}
  {058}},\ \Eprint {https://arxiv.org/abs/2401.12134} {arXiv:2401.12134
  [hep-th]} \BibitemShut {NoStop}%
\bibitem [{\citenamefont {Martin}\ \emph {et~al.}(2022)\citenamefont {Martin},
  \citenamefont {Micheli},\ and\ \citenamefont {Vennin}}]{Martin_2022}%
  \BibitemOpen
  \bibfield  {author} {\bibinfo {author} {\bibfnamefont {J.}~\bibnamefont
  {Martin}}, \bibinfo {author} {\bibfnamefont {A.}~\bibnamefont {Micheli}},\
  and\ \bibinfo {author} {\bibfnamefont {V.}~\bibnamefont {Vennin}},\
  }\bibfield  {title} {\bibinfo {title} {Discord and decoherence},\ }\href
  {https://doi.org/10.1088/1475-7516/2022/04/051} {\bibfield  {journal}
  {\bibinfo  {journal} {Journal of Cosmology and Astroparticle Physics}\
  }\textbf {\bibinfo {volume} {2022}}\bibinfo  {number} { (04)},\ \bibinfo
  {pages} {051}}\BibitemShut {NoStop}%
\bibitem [{\citenamefont {Calzetta}\ and\ \citenamefont
  {Hu}(1995)}]{Calzetta_1995}%
  \BibitemOpen
\bibfield  {number} {  }\bibfield  {author} {\bibinfo {author} {\bibfnamefont
  {E.}~\bibnamefont {Calzetta}}\ and\ \bibinfo {author} {\bibfnamefont {B.~L.}\
  \bibnamefont {Hu}},\ }\bibfield  {title} {\bibinfo {title} {Quantum
  fluctuations, decoherence of the mean field, and structure formation in the
  early universe},\ }\href {https://doi.org/10.1103/physrevd.52.6770}
  {\bibfield  {journal} {\bibinfo  {journal} {Physical Review D}\ }\textbf
  {\bibinfo {volume} {52}},\ \bibinfo {pages} {6770–6788} (\bibinfo {year}
  {1995})}\BibitemShut {NoStop}%
\bibitem [{\citenamefont {Maldacena}\ and\ \citenamefont
  {Pimentel}(2013)}]{maldacena_entanglement_2013}%
  \BibitemOpen
  \bibfield  {author} {\bibinfo {author} {\bibfnamefont {J.}~\bibnamefont
  {Maldacena}}\ and\ \bibinfo {author} {\bibfnamefont {G.~L.}\ \bibnamefont
  {Pimentel}},\ }\bibfield  {title} {\bibinfo {title} {Entanglement entropy in
  de {Sitter} space},\ }\href {https://doi.org/10.1007/JHEP02(2013)038}
  {\bibfield  {journal} {\bibinfo  {journal} {JHEP}\ }\textbf {\bibinfo
  {volume} {02}},\ \bibinfo {pages} {038}}\BibitemShut {NoStop}%
\bibitem [{\citenamefont {Maldacena}(2016)}]{Maldacena:2015bha}%
  \BibitemOpen
  \bibfield  {author} {\bibinfo {author} {\bibfnamefont {J.}~\bibnamefont
  {Maldacena}},\ }\bibfield  {title} {\bibinfo {title} {{A model with
  cosmological Bell inequalities}},\ }\href
  {https://doi.org/10.1002/prop.201500097} {\bibfield  {journal} {\bibinfo
  {journal} {Fortsch. Phys.}\ }\textbf {\bibinfo {volume} {64}},\ \bibinfo
  {pages} {10} (\bibinfo {year} {2016})},\ \Eprint
  {https://arxiv.org/abs/1508.01082} {arXiv:1508.01082 [hep-th]} \BibitemShut
  {NoStop}%
\bibitem [{\citenamefont {Nelson}(2016)}]{Nelson_2016}%
  \BibitemOpen
  \bibfield  {author} {\bibinfo {author} {\bibfnamefont {E.}~\bibnamefont
  {Nelson}},\ }\bibfield  {title} {\bibinfo {title} {Quantum decoherence during
  inflation from gravitational nonlinearities},\ }\href
  {https://doi.org/10.1088/1475-7516/2016/03/022} {\bibfield  {journal}
  {\bibinfo  {journal} {Journal of Cosmology and Astroparticle Physics}\
  }\textbf {\bibinfo {volume} {2016}}\bibinfo  {number} { (03)},\ \bibinfo
  {pages} {022–022}}\BibitemShut {NoStop}%
\bibitem [{\citenamefont {Martin}\ and\ \citenamefont
  {Vennin}(2018)}]{Martin_2018}%
  \BibitemOpen
\bibfield  {number} {  }\bibfield  {author} {\bibinfo {author} {\bibfnamefont
  {J.}~\bibnamefont {Martin}}\ and\ \bibinfo {author} {\bibfnamefont
  {V.}~\bibnamefont {Vennin}},\ }\bibfield  {title} {\bibinfo {title} {Non
  gaussianities from quantum decoherence during inflation},\ }\href
  {https://doi.org/10.1088/1475-7516/2018/06/037} {\bibfield  {journal}
  {\bibinfo  {journal} {Journal of Cosmology and Astroparticle Physics}\
  }\textbf {\bibinfo {volume} {2018}}\bibinfo  {number} { (06)},\ \bibinfo
  {pages} {037–037}}\BibitemShut {NoStop}%
\bibitem [{\citenamefont {Grain}\ and\ \citenamefont
  {Vennin}(2020)}]{Grain:2019vnq}%
  \BibitemOpen
\bibfield  {number} {  }\bibfield  {author} {\bibinfo {author} {\bibfnamefont
  {J.}~\bibnamefont {Grain}}\ and\ \bibinfo {author} {\bibfnamefont
  {V.}~\bibnamefont {Vennin}},\ }\bibfield  {title} {\bibinfo {title}
  {{Canonical transformations and squeezing formalism in cosmology}},\ }\href
  {https://doi.org/10.1088/1475-7516/2020/02/022} {\bibfield  {journal}
  {\bibinfo  {journal} {JCAP}\ }\textbf {\bibinfo {volume} {02}},\ \bibinfo
  {pages} {022}},\ \Eprint {https://arxiv.org/abs/1910.01916} {arXiv:1910.01916
  [astro-ph.CO]} \BibitemShut {NoStop}%
\bibitem [{\citenamefont {Brahma}\ \emph {et~al.}(2020)\citenamefont {Brahma},
  \citenamefont {Alaryani},\ and\ \citenamefont
  {Brandenberger}}]{Brahma:2020zpk}%
  \BibitemOpen
  \bibfield  {author} {\bibinfo {author} {\bibfnamefont {S.}~\bibnamefont
  {Brahma}}, \bibinfo {author} {\bibfnamefont {O.}~\bibnamefont {Alaryani}},\
  and\ \bibinfo {author} {\bibfnamefont {R.}~\bibnamefont {Brandenberger}},\
  }\bibfield  {title} {\bibinfo {title} {{Entanglement entropy of cosmological
  perturbations}},\ }\href {https://doi.org/10.1103/PhysRevD.102.043529}
  {\bibfield  {journal} {\bibinfo  {journal} {Phys. Rev. D}\ }\textbf {\bibinfo
  {volume} {102}},\ \bibinfo {pages} {043529} (\bibinfo {year} {2020})},\
  \Eprint {https://arxiv.org/abs/2005.09688} {arXiv:2005.09688 [hep-th]}
  \BibitemShut {NoStop}%
\bibitem [{\citenamefont {Brahma}\ \emph {et~al.}(2021)\citenamefont {Brahma},
  \citenamefont {Berera},\ and\ \citenamefont
  {Calder\'on-Figueroa}}]{Brahma:2021mng}%
  \BibitemOpen
  \bibfield  {author} {\bibinfo {author} {\bibfnamefont {S.}~\bibnamefont
  {Brahma}}, \bibinfo {author} {\bibfnamefont {A.}~\bibnamefont {Berera}},\
  and\ \bibinfo {author} {\bibfnamefont {J.}~\bibnamefont
  {Calder\'on-Figueroa}},\ }\bibfield  {title} {\bibinfo {title} {{Universal
  signature of quantum entanglement across cosmological distances}},\
  }\href@noop {} {\bibfield  {journal} {\bibinfo  {journal} {arXiv e-prints}\ }
  (\bibinfo {year} {2021})},\ \Eprint {https://arxiv.org/abs/2107.06910}
  {arXiv:2107.06910 [hep-th]} \BibitemShut {NoStop}%
\bibitem [{\citenamefont {Martin}\ and\ \citenamefont
  {Vennin}(2021{\natexlab{a}})}]{Martin:2021qkg}%
  \BibitemOpen
  \bibfield  {author} {\bibinfo {author} {\bibfnamefont {J.~e.}\ \bibnamefont
  {Martin}}\ and\ \bibinfo {author} {\bibfnamefont {V.}~\bibnamefont
  {Vennin}},\ }\bibfield  {title} {\bibinfo {title} {{Real-space entanglement
  in the Cosmic Microwave Background}},\ }\href
  {https://doi.org/10.1088/1475-7516/2021/10/036} {\bibfield  {journal}
  {\bibinfo  {journal} {JCAP}\ }\textbf {\bibinfo {volume} {10}},\ \bibinfo
  {pages} {036}},\ \Eprint {https://arxiv.org/abs/2106.15100} {arXiv:2106.15100
  [gr-qc]} \BibitemShut {NoStop}%
\bibitem [{\citenamefont {Agullo}\ \emph {et~al.}(2022)\citenamefont {Agullo},
  \citenamefont {Bonga},\ and\ \citenamefont
  {Ribes-Metidieri}}]{Agullo:2022ttg}%
  \BibitemOpen
  \bibfield  {author} {\bibinfo {author} {\bibfnamefont {I.}~\bibnamefont
  {Agullo}}, \bibinfo {author} {\bibfnamefont {B.}~\bibnamefont {Bonga}},\ and\
  \bibinfo {author} {\bibfnamefont {P.}~\bibnamefont {Ribes-Metidieri}},\
  }\bibfield  {title} {\bibinfo {title} {{Does inflation squeeze cosmological
  perturbations?}},\ }\href {https://doi.org/10.1088/1475-7516/2022/09/032}
  {\bibfield  {journal} {\bibinfo  {journal} {JCAP}\ }\textbf {\bibinfo
  {volume} {09}},\ \bibinfo {pages} {032}},\ \Eprint
  {https://arxiv.org/abs/2203.07066} {arXiv:2203.07066 [gr-qc]} \BibitemShut
  {NoStop}%
\bibitem [{\citenamefont {Espinosa-Portal\'es}\ and\ \citenamefont
  {Vennin}(2022)}]{Espinosa-Portales:2022yok}%
  \BibitemOpen
  \bibfield  {author} {\bibinfo {author} {\bibfnamefont {L.}~\bibnamefont
  {Espinosa-Portal\'es}}\ and\ \bibinfo {author} {\bibfnamefont
  {V.}~\bibnamefont {Vennin}},\ }\bibfield  {title} {\bibinfo {title}
  {{Real-space Bell inequalities in de Sitter}},\ }\href@noop {} {\bibfield
  {journal} {\bibinfo  {journal} {arXiv e-prints}\ } (\bibinfo {year}
  {2022})},\ \Eprint {https://arxiv.org/abs/2203.03505} {arXiv:2203.03505
  [quant-ph]} \BibitemShut {NoStop}%
\bibitem [{\citenamefont {Bhardwaj}\ \emph {et~al.}(2024)\citenamefont
  {Bhardwaj}, \citenamefont {Agullo}, \citenamefont {Kranas}, \citenamefont
  {Wilson},\ and\ \citenamefont {Sheehy}}]{Bhardwaj:2023squ}%
  \BibitemOpen
  \bibfield  {author} {\bibinfo {author} {\bibfnamefont {A.}~\bibnamefont
  {Bhardwaj}}, \bibinfo {author} {\bibfnamefont {I.}~\bibnamefont {Agullo}},
  \bibinfo {author} {\bibfnamefont {D.}~\bibnamefont {Kranas}}, \bibinfo
  {author} {\bibfnamefont {J.~H.}\ \bibnamefont {Wilson}},\ and\ \bibinfo
  {author} {\bibfnamefont {D.~E.}\ \bibnamefont {Sheehy}},\ }\bibfield  {title}
  {\bibinfo {title} {{Entanglement in an expanding toroidal Bose-Einstein
  condensate}},\ }\href {https://doi.org/10.1103/PhysRevA.109.013305}
  {\bibfield  {journal} {\bibinfo  {journal} {Phys. Rev. A}\ }\textbf {\bibinfo
  {volume} {109}},\ \bibinfo {pages} {013305} (\bibinfo {year} {2024})},\
  \Eprint {https://arxiv.org/abs/2307.07560} {arXiv:2307.07560
  [cond-mat.quant-gas]} \BibitemShut {NoStop}%
\bibitem [{\citenamefont {Lombardo}\ and\ \citenamefont
  {Nacir}(2005)}]{Lombardo_2005}%
  \BibitemOpen
  \bibfield  {author} {\bibinfo {author} {\bibfnamefont {F.~C.}\ \bibnamefont
  {Lombardo}}\ and\ \bibinfo {author} {\bibfnamefont {D.~L.}\ \bibnamefont
  {Nacir}},\ }\bibfield  {title} {\bibinfo {title} {Decoherence during
  inflation: The generation of classical inhomogeneities},\ }\bibfield
  {journal} {\bibinfo  {journal} {Physical Review D}\ }\textbf {\bibinfo
  {volume} {72}},\ \href {https://doi.org/10.1103/physrevd.72.063506}
  {10.1103/physrevd.72.063506} (\bibinfo {year} {2005})\BibitemShut {NoStop}%
\bibitem [{\citenamefont {Burgess}\ \emph {et~al.}(2008)\citenamefont
  {Burgess}, \citenamefont {Holman},\ and\ \citenamefont
  {Hoover}}]{Burgess_2008}%
  \BibitemOpen
  \bibfield  {author} {\bibinfo {author} {\bibfnamefont {C.~P.}\ \bibnamefont
  {Burgess}}, \bibinfo {author} {\bibfnamefont {R.}~\bibnamefont {Holman}},\
  and\ \bibinfo {author} {\bibfnamefont {D.}~\bibnamefont {Hoover}},\
  }\bibfield  {title} {\bibinfo {title} {Decoherence of inflationary primordial
  fluctuations},\ }\bibfield  {journal} {\bibinfo  {journal} {Physical Review
  D}\ }\textbf {\bibinfo {volume} {77}},\ \href
  {https://doi.org/10.1103/physrevd.77.063534} {10.1103/physrevd.77.063534}
  (\bibinfo {year} {2008})\BibitemShut {NoStop}%
\bibitem [{\citenamefont {Burgess}\ \emph {et~al.}(2014)\citenamefont
  {Burgess}, \citenamefont {Holman}, \citenamefont {Tasinato},\ and\
  \citenamefont {Williams}}]{burgess2014efthorizonstochasticinflation}%
  \BibitemOpen
  \bibfield  {author} {\bibinfo {author} {\bibfnamefont {C.~P.}\ \bibnamefont
  {Burgess}}, \bibinfo {author} {\bibfnamefont {R.}~\bibnamefont {Holman}},
  \bibinfo {author} {\bibfnamefont {G.}~\bibnamefont {Tasinato}},\ and\
  \bibinfo {author} {\bibfnamefont {M.}~\bibnamefont {Williams}},\ }\href
  {https://arxiv.org/abs/1408.5002} {\bibinfo {title} {{EFT Beyond the Horizon:
  Stochastic Inflation and How Primordial Quantum Fluctuations Go Classical}}}
  (\bibinfo {year} {2014}),\ \Eprint {https://arxiv.org/abs/1408.5002}
  {arXiv:1408.5002 [hep-th]} \BibitemShut {NoStop}%
\bibitem [{\citenamefont {Burgess}\ \emph {et~al.}(2023)\citenamefont
  {Burgess}, \citenamefont {Holman}, \citenamefont {Kaplanek}, \citenamefont
  {Martin},\ and\ \citenamefont {Vennin}}]{Burgess_2023}%
  \BibitemOpen
  \bibfield  {author} {\bibinfo {author} {\bibfnamefont {C.}~\bibnamefont
  {Burgess}}, \bibinfo {author} {\bibfnamefont {R.}~\bibnamefont {Holman}},
  \bibinfo {author} {\bibfnamefont {G.}~\bibnamefont {Kaplanek}}, \bibinfo
  {author} {\bibfnamefont {J.}~\bibnamefont {Martin}},\ and\ \bibinfo {author}
  {\bibfnamefont {V.}~\bibnamefont {Vennin}},\ }\bibfield  {title} {\bibinfo
  {title} {Minimal decoherence from inflation},\ }\href
  {https://doi.org/10.1088/1475-7516/2023/07/022} {\bibfield  {journal}
  {\bibinfo  {journal} {Journal of Cosmology and Astroparticle Physics}\
  }\textbf {\bibinfo {volume} {2023}}\bibinfo  {number} { (07)},\ \bibinfo
  {pages} {022}}\BibitemShut {NoStop}%
\bibitem [{\citenamefont {Martin}\ and\ \citenamefont
  {Vennin}(2021{\natexlab{b}})}]{Martin:2021xml}%
  \BibitemOpen
\bibfield  {number} {  }\bibfield  {author} {\bibinfo {author} {\bibfnamefont
  {J.}~\bibnamefont {Martin}}\ and\ \bibinfo {author} {\bibfnamefont
  {V.}~\bibnamefont {Vennin}},\ }\bibfield  {title} {\bibinfo {title}
  {{Real-space entanglement of quantum fields}},\ }\href
  {https://doi.org/10.1103/PhysRevD.104.085012} {\bibfield  {journal} {\bibinfo
   {journal} {Phys. Rev. D}\ }\textbf {\bibinfo {volume} {104}},\ \bibinfo
  {pages} {085012} (\bibinfo {year} {2021}{\natexlab{b}})},\ \Eprint
  {https://arxiv.org/abs/2106.14575} {arXiv:2106.14575 [hep-th]} \BibitemShut
  {NoStop}%
\bibitem [{\citenamefont {K}\ \emph {et~al.}(2024)\citenamefont {K},
  \citenamefont {Barman},\ and\ \citenamefont {Kothawala}}]{K:2023oon}%
  \BibitemOpen
  \bibfield  {author} {\bibinfo {author} {\bibfnamefont {H.}~\bibnamefont {K}},
  \bibinfo {author} {\bibfnamefont {S.}~\bibnamefont {Barman}},\ and\ \bibinfo
  {author} {\bibfnamefont {D.}~\bibnamefont {Kothawala}},\ }\bibfield  {title}
  {\bibinfo {title} {{Universal role of curvature in vacuum entanglement}},\
  }\href {https://doi.org/10.1103/PhysRevD.109.065017} {\bibfield  {journal}
  {\bibinfo  {journal} {Phys. Rev. D}\ }\textbf {\bibinfo {volume} {109}},\
  \bibinfo {pages} {065017} (\bibinfo {year} {2024})},\ \Eprint
  {https://arxiv.org/abs/2311.15019} {arXiv:2311.15019 [gr-qc]} \BibitemShut
  {NoStop}%
\bibitem [{\citenamefont {{Planck Collaboration}}\ \emph
  {et~al.}(2020)\citenamefont {{Planck Collaboration}}, \citenamefont {{Akrami,
  Y.}}, \citenamefont {{Arroja, F.}} \emph {et~al.}}]{refId0}%
  \BibitemOpen
  \bibfield  {author} {\bibinfo {author} {\bibnamefont {{Planck
  Collaboration}}}, \bibinfo {author} {\bibnamefont {{Akrami, Y.}}}, \bibinfo
  {author} {\bibnamefont {{Arroja, F.}}}, \emph {et~al.},\ }\bibfield  {title}
  {\bibinfo {title} {Planck 2018 results - ix. constraints on primordial
  non-gaussianity},\ }\href {https://doi.org/10.1051/0004-6361/201935891}
  {\bibfield  {journal} {\bibinfo  {journal} {A\&A}\ }\textbf {\bibinfo
  {volume} {641}},\ \bibinfo {pages} {A9} (\bibinfo {year} {2020})}\BibitemShut
  {NoStop}%
\bibitem [{\citenamefont {Bianchi}\ and\ \citenamefont
  {Satz}(2019)}]{bianchi_entropy_2019}%
  \BibitemOpen
  \bibfield  {author} {\bibinfo {author} {\bibfnamefont {E.}~\bibnamefont
  {Bianchi}}\ and\ \bibinfo {author} {\bibfnamefont {A.}~\bibnamefont {Satz}},\
  }\bibfield  {title} {\bibinfo {title} {Entropy of a subalgebra of observables
  and the geometric entanglement entropy},\ }\href
  {https://doi.org/10.1103/PhysRevD.99.085001} {\bibfield  {journal} {\bibinfo
  {journal} {Physical Review D}\ }\textbf {\bibinfo {volume} {99}},\ \bibinfo
  {pages} {085001} (\bibinfo {year} {2019})},\ \bibinfo {note} {arXiv:
  1901.06454}\BibitemShut {NoStop}%
\bibitem [{\citenamefont {Agullo}\ \emph {et~al.}(2023)\citenamefont {Agullo},
  \citenamefont {Bonga}, \citenamefont {Ribes-Metidieri}, \citenamefont
  {Kranas},\ and\ \citenamefont {Nadal-Gisbert}}]{Agullo:2023fnp}%
  \BibitemOpen
  \bibfield  {author} {\bibinfo {author} {\bibfnamefont {I.}~\bibnamefont
  {Agullo}}, \bibinfo {author} {\bibfnamefont {B.}~\bibnamefont {Bonga}},
  \bibinfo {author} {\bibfnamefont {P.}~\bibnamefont {Ribes-Metidieri}},
  \bibinfo {author} {\bibfnamefont {D.}~\bibnamefont {Kranas}},\ and\ \bibinfo
  {author} {\bibfnamefont {S.}~\bibnamefont {Nadal-Gisbert}},\ }\bibfield
  {title} {\bibinfo {title} {{How ubiquitous is entanglement in quantum field
  theory?}},\ }\href {https://doi.org/10.1103/PhysRevD.108.085005} {\bibfield
  {journal} {\bibinfo  {journal} {Phys. Rev. D}\ }\textbf {\bibinfo {volume}
  {108}},\ \bibinfo {pages} {085005} (\bibinfo {year} {2023})},\ \Eprint
  {https://arxiv.org/abs/2302.13742} {arXiv:2302.13742 [quant-ph]} \BibitemShut
  {NoStop}%
\bibitem [{\citenamefont {Perche}\ \emph
  {et~al.}(2024{\natexlab{a}})\citenamefont {Perche}, \citenamefont
  {Polo-G\'omez}, \citenamefont {de~S.~L.~Torres},\ and\ \citenamefont
  {Mart\'\i{}n-Mart\'\i{}nez}}]{Perche:2023nde}%
  \BibitemOpen
  \bibfield  {author} {\bibinfo {author} {\bibfnamefont {T.~R.}\ \bibnamefont
  {Perche}}, \bibinfo {author} {\bibfnamefont {J.}~\bibnamefont
  {Polo-G\'omez}}, \bibinfo {author} {\bibfnamefont {B.}~\bibnamefont
  {de~S.~L.~Torres}},\ and\ \bibinfo {author} {\bibfnamefont {E.}~\bibnamefont
  {Mart\'\i{}n-Mart\'\i{}nez}},\ }\bibfield  {title} {\bibinfo {title} {{Fully
  relativistic entanglement harvesting}},\ }\href
  {https://doi.org/10.1103/PhysRevD.109.045018} {\bibfield  {journal} {\bibinfo
   {journal} {Phys. Rev. D}\ }\textbf {\bibinfo {volume} {109}},\ \bibinfo
  {pages} {045018} (\bibinfo {year} {2024}{\natexlab{a}})},\ \Eprint
  {https://arxiv.org/abs/2310.18432} {arXiv:2310.18432 [quant-ph]} \BibitemShut
  {NoStop}%
\bibitem [{\citenamefont {Perche}\ \emph
  {et~al.}(2024{\natexlab{b}})\citenamefont {Perche}, \citenamefont
  {Polo-G\'omez}, \citenamefont {de~S.~L.~Torres},\ and\ \citenamefont
  {Mart\'\i{}n-Mart\'\i{}nez}}]{Perche:2023lwo}%
  \BibitemOpen
  \bibfield  {author} {\bibinfo {author} {\bibfnamefont {T.~R.}\ \bibnamefont
  {Perche}}, \bibinfo {author} {\bibfnamefont {J.}~\bibnamefont
  {Polo-G\'omez}}, \bibinfo {author} {\bibfnamefont {B.}~\bibnamefont
  {de~S.~L.~Torres}},\ and\ \bibinfo {author} {\bibfnamefont {E.}~\bibnamefont
  {Mart\'\i{}n-Mart\'\i{}nez}},\ }\bibfield  {title} {\bibinfo {title}
  {{Particle detectors from localized quantum field theories}},\ }\href
  {https://doi.org/10.1103/PhysRevD.109.045013} {\bibfield  {journal} {\bibinfo
   {journal} {Phys. Rev. D}\ }\textbf {\bibinfo {volume} {109}},\ \bibinfo
  {pages} {045013} (\bibinfo {year} {2024}{\natexlab{b}})},\ \Eprint
  {https://arxiv.org/abs/2308.11698} {arXiv:2308.11698 [quant-ph]} \BibitemShut
  {NoStop}%
\bibitem [{\citenamefont {Bunch}\ and\ \citenamefont
  {Davies}(1978)}]{Bunch:1978yq}%
  \BibitemOpen
  \bibfield  {author} {\bibinfo {author} {\bibfnamefont {T.~S.}\ \bibnamefont
  {Bunch}}\ and\ \bibinfo {author} {\bibfnamefont {P.~C.~W.}\ \bibnamefont
  {Davies}},\ }\bibfield  {title} {\bibinfo {title} {{Quantum Field Theory in
  de Sitter Space: Renormalization by Point Splitting}},\ }\href
  {https://doi.org/10.1098/rspa.1978.0060} {\bibfield  {journal} {\bibinfo
  {journal} {Proc. Roy. Soc. Lond. A}\ }\textbf {\bibinfo {volume} {360}},\
  \bibinfo {pages} {117} (\bibinfo {year} {1978})}\BibitemShut {NoStop}%
\bibitem [{\citenamefont {Wald}(1995)}]{Wald:1995yp}%
  \BibitemOpen
  \bibfield  {author} {\bibinfo {author} {\bibfnamefont {R.~M.}\ \bibnamefont
  {Wald}},\ }\href@noop {} {\emph {\bibinfo {title} {{Quantum Field Theory in
  Curved Space-Time and Black Hole Thermodynamics}}}},\ Chicago Lectures in
  Physics\ (\bibinfo  {publisher} {University of Chicago Press},\ \bibinfo
  {address} {Chicago, IL},\ \bibinfo {year} {1995})\BibitemShut {NoStop}%
\bibitem [{\citenamefont {Ribes-Metidieri}\ \emph {et~al.}()\citenamefont
  {Ribes-Metidieri}, \citenamefont {Agullo},\ and\ \citenamefont
  {Bonga}}]{longpaper}%
  \BibitemOpen
  \bibfield  {author} {\bibinfo {author} {\bibfnamefont {P.}~\bibnamefont
  {Ribes-Metidieri}}, \bibinfo {author} {\bibfnamefont {I.}~\bibnamefont
  {Agullo}},\ and\ \bibinfo {author} {\bibfnamefont {B.}~\bibnamefont
  {Bonga}},\ }\bibfield  {title} {\bibinfo {title} {{Entanglement in de Sitter
  spacetime}},\ }\href@noop {} {\bibinfo  {journal} {{In preparation}}\
  }\BibitemShut {NoStop}%
\bibitem [{\citenamefont {Liddle}\ and\ \citenamefont
  {Leach}(2003)}]{Liddle_2003}%
  \BibitemOpen
\bibfield  {journal} {  }\bibfield  {author} {\bibinfo {author} {\bibfnamefont
  {A.~R.}\ \bibnamefont {Liddle}}\ and\ \bibinfo {author} {\bibfnamefont
  {S.~M.}\ \bibnamefont {Leach}},\ }\bibfield  {title} {\bibinfo {title} {How
  long before the end of inflation were observable perturbations produced?},\
  }\bibfield  {journal} {\bibinfo  {journal} {Physical Review D}\ }\textbf
  {\bibinfo {volume} {68}},\ \href {https://doi.org/10.1103/physrevd.68.103503}
  {10.1103/physrevd.68.103503} (\bibinfo {year} {2003})\BibitemShut {NoStop}%
\bibitem [{\citenamefont {Bahouri}\ \emph {et~al.}(2011)\citenamefont
  {Bahouri}, \citenamefont {Chemin},\ and\ \citenamefont
  {Danchin}}]{bahouri_fourier_2011}%
  \BibitemOpen
  \bibfield  {author} {\bibinfo {author} {\bibfnamefont {H.}~\bibnamefont
  {Bahouri}}, \bibinfo {author} {\bibfnamefont {J.-Y.}\ \bibnamefont
  {Chemin}},\ and\ \bibinfo {author} {\bibfnamefont {R.}~\bibnamefont
  {Danchin}},\ }\href@noop {} {\emph {\bibinfo {title} {Fourier {Analysis} and
  {Nonlinear} {Partial} {Differential} {Equations}}}}\ (\bibinfo  {publisher}
  {Springer Science \& Business Media},\ \bibinfo {year} {2011})\BibitemShut
  {NoStop}%
\bibitem [{\citenamefont {Vidal}\ and\ \citenamefont {Werner}(2002)}]{vidal02}%
  \BibitemOpen
  \bibfield  {author} {\bibinfo {author} {\bibfnamefont {G.}~\bibnamefont
  {Vidal}}\ and\ \bibinfo {author} {\bibfnamefont {R.~F.}\ \bibnamefont
  {Werner}},\ }\bibfield  {title} {\bibinfo {title} {Computable measure of
  entanglement},\ }\href {https://doi.org/10.1103/PhysRevA.65.032314}
  {\bibfield  {journal} {\bibinfo  {journal} {Physical Review A}\ }\textbf
  {\bibinfo {volume} {65}},\ \bibinfo {pages} {032314} (\bibinfo {year}
  {2002})}\BibitemShut {NoStop}%
\bibitem [{\citenamefont {Plenio}(2005)}]{plenio05}%
  \BibitemOpen
  \bibfield  {author} {\bibinfo {author} {\bibfnamefont {M.~B.}\ \bibnamefont
  {Plenio}},\ }\bibfield  {title} {\bibinfo {title} {Logarithmic negativity: A
  full entanglement monotone that is not convex},\ }\href
  {https://doi.org/10.1103/PhysRevLett.95.090503} {\bibfield  {journal}
  {\bibinfo  {journal} {Physical Review Letters}\ }\textbf {\bibinfo {volume}
  {95}},\ \bibinfo {pages} {090503} (\bibinfo {year} {2005})}\BibitemShut
  {NoStop}%
\bibitem [{\citenamefont {Werner}\ and\ \citenamefont
  {Wolf}(2001)}]{werner01bound}%
  \BibitemOpen
  \bibfield  {author} {\bibinfo {author} {\bibfnamefont {R.~F.}\ \bibnamefont
  {Werner}}\ and\ \bibinfo {author} {\bibfnamefont {M.~M.}\ \bibnamefont
  {Wolf}},\ }\bibfield  {title} {\bibinfo {title} {Bound entangled gaussian
  states},\ }\href {https://doi.org/10.1103/PhysRevLett.86.3658} {\bibfield
  {journal} {\bibinfo  {journal} {Physical review letters}\ }\textbf {\bibinfo
  {volume} {86}},\ \bibinfo {pages} {3658} (\bibinfo {year}
  {2001})}\BibitemShut {NoStop}%
\bibitem [{\citenamefont {Botero}\ and\ \citenamefont
  {Reznik}(2003)}]{Botero_2003}%
  \BibitemOpen
  \bibfield  {author} {\bibinfo {author} {\bibfnamefont {A.}~\bibnamefont
  {Botero}}\ and\ \bibinfo {author} {\bibfnamefont {B.}~\bibnamefont
  {Reznik}},\ }\bibfield  {title} {\bibinfo {title} {Modewise entanglement of
  gaussian states},\ }\bibfield  {journal} {\bibinfo  {journal} {Physical
  Review A}\ }\textbf {\bibinfo {volume} {67}},\ \href
  {https://doi.org/10.1103/physreva.67.052311} {10.1103/physreva.67.052311}
  (\bibinfo {year} {2003})\BibitemShut {NoStop}%
\bibitem [{\citenamefont {Hotta}\ \emph {et~al.}(2015)\citenamefont {Hotta},
  \citenamefont {Schützhold},\ and\ \citenamefont {Unruh}}]{hotta2015partner}%
  \BibitemOpen
  \bibfield  {author} {\bibinfo {author} {\bibfnamefont {M.}~\bibnamefont
  {Hotta}}, \bibinfo {author} {\bibfnamefont {R.}~\bibnamefont {Schützhold}},\
  and\ \bibinfo {author} {\bibfnamefont {W.}~\bibnamefont {Unruh}},\
  }\href@noop {} {\bibinfo {title} {On the partner particles for moving mirror
  radiation and black hole evaporation}} (\bibinfo {year} {2015}),\ \Eprint
  {https://arxiv.org/abs/1503.06109} {arXiv:1503.06109 [gr-qc]} \BibitemShut
  {NoStop}%
\bibitem [{\citenamefont {Trevison}\ \emph {et~al.}(2019)\citenamefont
  {Trevison}, \citenamefont {Yamaguchi},\ and\ \citenamefont
  {Hotta}}]{Trevison_2019}%
  \BibitemOpen
  \bibfield  {author} {\bibinfo {author} {\bibfnamefont {J.}~\bibnamefont
  {Trevison}}, \bibinfo {author} {\bibfnamefont {K.}~\bibnamefont
  {Yamaguchi}},\ and\ \bibinfo {author} {\bibfnamefont {M.}~\bibnamefont
  {Hotta}},\ }\bibfield  {title} {\bibinfo {title} {Spatially overlapped
  partners in quantum field theory},\ }\href
  {https://doi.org/10.1088/1751-8121/ab065b} {\bibfield  {journal} {\bibinfo
  {journal} {Journal of Physics A: Mathematical and Theoretical}\ }\textbf
  {\bibinfo {volume} {52}},\ \bibinfo {pages} {125402} (\bibinfo {year}
  {2019})}\BibitemShut {NoStop}%
\bibitem [{\citenamefont {Agullo}\ \emph {et~al.}()\citenamefont {Agullo},
  \citenamefont {Martin-Martinez}, \citenamefont {Nadal-Gisbert}, \citenamefont
  {Ribes-Metidieri},\ and\ \citenamefont {Yamaguchi}}]{partnerformula}%
  \BibitemOpen
  \bibfield  {author} {\bibinfo {author} {\bibfnamefont {I.}~\bibnamefont
  {Agullo}}, \bibinfo {author} {\bibfnamefont {E.}~\bibnamefont
  {Martin-Martinez}}, \bibinfo {author} {\bibfnamefont {S.}~\bibnamefont
  {Nadal-Gisbert}}, \bibinfo {author} {\bibfnamefont {P.}~\bibnamefont
  {Ribes-Metidieri}},\ and\ \bibinfo {author} {\bibfnamefont {K.}~\bibnamefont
  {Yamaguchi}},\ }\href@noop {} {\bibinfo  {journal} {{In preparation}}\
  }\BibitemShut {NoStop}%
\bibitem [{\citenamefont {Steeg}\ and\ \citenamefont
  {Menicucci}(2009)}]{Steeg_2009}%
  \BibitemOpen
\bibfield  {journal} {  }\bibfield  {author} {\bibinfo {author} {\bibfnamefont
  {G.~V.}\ \bibnamefont {Steeg}}\ and\ \bibinfo {author} {\bibfnamefont
  {N.~C.}\ \bibnamefont {Menicucci}},\ }\bibfield  {title} {\bibinfo {title}
  {Entangling power of an expanding universe},\ }\bibfield  {journal} {\bibinfo
   {journal} {Physical Review D}\ }\textbf {\bibinfo {volume} {79}},\ \href
  {https://doi.org/10.1103/physrevd.79.044027} {10.1103/physrevd.79.044027}
  (\bibinfo {year} {2009})\BibitemShut {NoStop}%
\end{thebibliography}%

\end{document}